\documentclass[[10pt,journal,compsoc]{IEEEtran}
\usepackage{cite}
\usepackage[numbers]{natbib}

\ifCLASSINFOpdf
  \usepackage[pdftex]{graphicx}
\else
\fi
\usepackage{amsmath,amssymb,amsfonts,amsthm}
\usepackage{pifont}
\newcommand{\cmark}{\ding{51}}%
\newcommand{\xmark}{\ding{55}}%

\usepackage{algorithm}
\usepackage{algpseudocode}

\newtheorem{assumption}{Assumption}[section]
\newtheorem{theorem}{Theorem}[section]
\newtheorem{lemma}[theorem]{Lemma}

\newtheorem{corollary}{Corollary}[theorem]

\newcommand{\algorithmicdefinition}{\textbf{Definitions:}}
\newcommand{\Definition}{\item[\algorithmicdefinition]}
\newcommand{\algorithmicdefta}{\textbf{DeFTA():}}
\newcommand{\DeFTA}{\item[\algorithmicdefta]}

\newcommand{\Blank}{\textbf{}}
\newcommand{\algorithmiconconnect}{\textbf{OnConnect($j$):}}
\newcommand{\OnConnect}{\item[\algorithmiconconnect]}
\newcommand{\algorithmiconreceive}{\textbf{OnReceive($|\mathcal{D}_j|, d_j, w_j$):}}
\newcommand{\OnReceive}{\item[\algorithmiconreceive]}

\newcommand{\algorithmicvarphi}{\textbf{$\varphi(\{ |\mathcal{D}_j|, d_j, \tilde{w}_j | j \in \mathcal{S}_i^t\})$:}}
\newcommand{\funcvarphi}{\item[\algorithmicvarphi]}

\newcommand{\algorithmicphi}{\textbf{$\phi(\mathbf{c}_{i}^t, \tilde{w}_i^t)$:}}
\newcommand{\funcphi}{\item[\algorithmicphi]}

\ifCLASSOPTIONcompsoc
 \usepackage[caption=false,font=normalsize,labelfont=sf,textfont=sf]{subfig}
\else
 \usepackage[caption=false,font=footnotesize]{subfig}
\fi

\usepackage{stfloats}
\usepackage{url}

\usepackage{caption}
\usepackage{tabularx}
\usepackage{multirow}
\usepackage{booktabs}
\usepackage{mathtools}
\usepackage{color,soul}

\begin{document}
%
\title{DeFTA: A Plug-and-Play Decentralized Replacement for FedAvg}
%
%
%

\author{Yuhao~Zhou,~Minjia~Shi,~Yuxin~Tian,~Qing~Ye,~and~Jiancheng Lv,~\IEEEmembership{Member,~IEEE}
\thanks{This work is supported in part by the National Key Research and Development Program of China under Contract 2017YFB1002201, in part by the National Natural Science Fund for Distinguished Young Scholar under Grant 61625204, and in part by the State Key Program of the National Science Foundation of China under Grant 61836006.}
\thanks{Yuhao~Zhou~(e-mail:~sooptq@gmail.com), Minjia~Shi~(e-mail:~3101ihs@gmail.com), Yuxin~Tian~(e-mail:~cs.yuxintian@outlook.com), Qing~Ye~(e-mail:~fuyeking@stu.scu.edu.cn), and Jiancheng~Lv~(e-mail:~lvjiancheng@scu.edu.cn) are with the College of Computer Science, Sichuan University, China.}
}

%
%

\markboth{Journal of \LaTeX\ Class Files,~Vol.~14, No.~8, August~2015}%
{Zhou \MakeLowercase{\textit{et al.}}: Communication-Efficient Federated Learning with Compensated Overlap-FedAvg}
%



\IEEEtitleabstractindextext{%
\begin{abstract}
Federated learning (FL) is identified as a crucial enabler for large-scale distributed machine learning (ML) without the need for local raw dataset sharing, substantially reducing privacy concerns and alleviating the isolated data problem. In reality, the prosperity of FL is largely due to a centralized framework called FedAvg~\cite{mcmahan2016federated}, in which workers are in charge of model training and servers are in control of model aggregation. However, FedAvg's centralized worker-server architecture has raised new concerns, be it the low scalability of the cluster, the risk of data leakage, and the failure or even defection of the central server. To overcome these problems, we propose \textbf{De}centralized \textbf{F}ederated \textbf{T}rusted \textbf{A}veraging (DeFTA), a decentralized FL framework that serves as \textit{a plug-and-play replacement for FedAvg}, instantly bringing better security, scalability, and fault-tolerance to the federated learning process after installation. In principle, it fundamentally resolves the above-mentioned issues from an architectural perspective without compromises or tradeoffs, primarily consisting of a new model aggregating formula with theoretical performance analysis, and a decentralized trust system (DTS) to greatly improve system robustness. Note that since DeFTA is an alternative to FedAvg at the framework level, \textit{prevalent algorithms published for FedAvg can be also utilized in DeFTA with ease}. Extensive experiments on six datasets and six basic models suggest that DeFTA not only has comparable performance with FedAvg in a more realistic setting, but also achieves great resilience even when \textit{66\%} of workers are malicious. Furthermore, we also present an asynchronous variant of DeFTA to endow it with more powerful usability. 
\end{abstract}

}

\maketitle



%
\IEEEpeerreviewmaketitle

\section{Introduction}
%
%
%
%
\IEEEPARstart{F}ederated learning (FL) was proposed in \cite{mcmahan2017communication, mcmahan2016federated, li2018federated, zhao2018federated} to mainly tackle isolated data islands problem, allowing centralized-data-free distributed ML, and having been applied in industrial fields to guard data privacy, \textit{e.g.,} keyboard prediction~\cite{hard2018federated}, disease diagnose~\cite{brisimi2018federated}, and vehicular communications~\cite{samarakoon2019distributed}. Until now, the most commonly used FL framework FedAvg~\cite{mcmahan2016federated} utilizes classical worker-server architecture, where dedicated central servers are employed to orchestrate the entire training process, asking workers to upload their trained local models instead of local raw data. Yet, as the usage scenarios of FL expand, several weaknesses of FedAvg are gradually exposed, be it massive communication load with vast workers~\cite{konevcny2016federated}, potential failure of central servers~\cite{kairouz2019advances}, or even possible dataset leakage from exchanged local models~\cite{zhu2020deep}. These weaknesses contribute to the hindering of the wider adoption of FedAvg.

As a matter of fact, various algorithms were proposed to alleviate the above-mentioned issues case by case. \cite{zhou2021communication} was proposed to apply overlapping training with data compensation to the FL process, efficiently dispersing the communication pressure with minimal model performance decrease. \cite{bonawitz2019towards} built a scalable production system for FL to improve system reliability, including complicated architecture design and data pipeline. \cite{mcmahan2017learning, bonawitz2017practical} introduced Differential Privacy (DP) and Secure Aggregation (SecAgg) respectively to prevent central servers from restoring workers' local raw dataset, but still the former approach affects the model performance~\cite{bagdasaryan2019differential}, and the latter algorithm considerably aggravated the communication load of the entire cluster~\cite{bonawitz2019federated}, making FL difficult to scale up. As it can be seen, these innovative methods are essentially making tradeoffs in different aspects (\textit{i.e.,} model performance, system reliability, communication load, etc.), addressing some issues while introducing some others. In fact, people are realizing that although the worker-server architecture was considered the most popular at that time (\textit{i.e.,} when FedAvg was firstly proposed), a majority of common problems in practical FL is also occurred by this architecture as well~\cite{kairouz2019advances} (\textit{i.e.} unbalanced communication load, low fault-tolerance, etc.), and thereby blocking FL's further developments. Thus, a more secure and robust FL framework needs to be came up with as a replacement for FedAvg.

Consequently, several works~\cite{roy2019braintorrent, hu2019decentralized, lalitha2019peer, kim2019blockchained} had replaced the default worker-server architecture in centralized FL in a peer-to-peer (p2p) manner~\cite{saroiu2003measuring} for better scalability and ownership~\cite{kairouz2019advances}. However, these methods just laser-focus on one of the issues mentioned above as well, and even raise new problems, \textit{i.e.,} performance degradation~\cite{roy2019braintorrent, hu2019decentralized} and impracticality~\cite{kim2019blockchained, roy2019braintorrent}. In other words, a practical decentralized system is full-fill with sparse connected and untrustworthy peers, without a synchronized global clock, as is one of the main concerns of this paper. Hence, we pursue to practically scale centralized FL up to real-world decentralized architecture without performance compromise.

Motivated by this, we propose a decentralized FL framework with more practicability compared to previous proposals, named \textbf{De}centralized \textbf{F}ederated \textbf{T}rusted \textbf{A}veraging (DeFTA), aiming to address both unresolved problems in the centralized FL (\textit{i.e.,} massive communication load~\cite{konevcny2016federated}, failure of central servers~\cite{kairouz2019advances}, dataset leakage~\cite{zhu2020deep}) and new raising real-world challenges in decentralized FL (\textit{i.e.,} sparse connection, malicious peers, and asynchronism). In other words, \textit{DeFTA serves as a plug-and-play decentralized replacement for FedAvg for instantly improvements on model performance, data security, and system reliability, with high flexibility, minimal intrusion to the underlying implemented algorithms within the framework, and is fully compatible with all previous algorithms for FedAvg (\textit{i.e.,} DP~\cite{mcmahan2017learning}, SecAgg~\cite{bonawitz2017practical}), FedAdam~\cite{reddi2020adaptive}, etc.}. To be more specific, unlike directly using the ratio of dataset size as model aggregating weights, we novelly take workers' outdegrees in the connected graph into account to rectify the averaging bias caused by decentralized broadcasting, and provide the theoretical convergence analysis. Second, assuming all peers are unreliable by default, we introduce a decentralized trust system (DTS) which could isolate malicious actors for better fault tolerance and robustness. Furthermore, we also proposed an asynchronous variant of DeFTA, namely, AsyncDeFTA, to remove the barriers of synchronization. It could be one of the first attempts to deliver an integrated decentralized FL framework as a drop-in replacement for classical FedAvg for wildness deployment. Extensive experiments also demonstrate comparable performance with centralized FL and stability under malicious decentralized environments, validating the effectiveness of DeFTA.

The contributions of this proposal are the followings:

\begin{enumerate}
    \item We revise the previous model aggregating formula in a decentralized FL setting by taking workers' outdegrees into account. It successfully rectifies the averaging bias caused by decentralized broadcasting.
    \item We bring selfish assumptions (\textit{i.e.,} workers are assumed to be untrustworthy by default) to the decentralized FL setting, making the system robust to adversarial actors.
    \item We present an integrated decentralized FL framework that features a drop-in replacement for classical FedAvg for instant improvements on model performance, data security, system reliability, and more. It keeps full compatibility with previous algorithms proposed for FedAvg (\textit{i.e.,} DP~\cite{mcmahan2017learning}, SecAgg~\cite{bonawitz2017practical}), FedAdam~\cite{reddi2020adaptive}, etc.) so that they can be used combinedly for even greater enhancements. To our best knowledge, it could be the first attempt in the decentralized FL realm to satisfy these characteristics at once.
    \item We demonstrate extensive experiments conducted on six model structures and six datasets. Their results suggest marginal improvements over other decentralized FL attempts and strong system stability under harsh environments, validating the effectiveness of DeFTA.
\end{enumerate}

\section{Literature Review}
\subsection{Decentralized Optimization}
Decentralized optimization~\cite{tsitsiklis1986distributed} represents schemes that help the model be trained decentrally for better scalability and ownership~\cite{kairouz2019advances} compared to centralized optimization (\textit{e.g.,} Parameter Server~\cite{li2014scaling}). The overall workflow of decentralized optimization can be summaries as a two-step strategy, commonly known as Combine-Then-Adapt (CTA) Diffusion Strategy~\cite{sayed2014adaptation}: 1. each worker takes the weighted average of data sent from other workers that it connects to (\textit{i.e.,} its peers) as its new local data. 2. each worker updates its local data and sends it to its peers. By repeating the above 2 steps, the optimization of each worker's local data can be collaboratively achieved. However, it is obvious that weights for averaging in step 1 are crucial to the general optimizing performance, and its optimal value under decentralized FL settings is still an open problem for discussion.

In fact, the application of the decentralized optimization in ML can be tracked back to~\cite{tsitsiklis1984problems}, and then followed by an abundant of exceptional works for collaboratively model training~\cite{colin2016gossip, vanhaesebrouck2017decentralized, tang2018d, bellet2018personalized, koloskova2019decentralized, elgabli2020gadmm}. However, these works focus on the decentralization of distributed ML, and therefore are commonly suffered from high communication load or un-secured local datasets, making them fail to be directly applied to the FL realm, where both privacy, latency, and model performance are required at the same time. In our case, we focus on delivering a decentralized FL framework that best suits the needs of FL, addressing the above concerns from an architecture perspective. The general workflow of our framework follows the classical CTA strategy for flexibility so that it is very easy to extend its functionalities, with derived optimal weights for averaging under FL settings to guarantee high optimizing performance.

\subsection{Consensus Protocol}
Consensus protocols ensure operations in a distributed system are always stable~\cite{zhang2020analysis} even in the presence of a number of faulty processes (\textit{i.e.,} adversarial attack, offline workers, etc.). A common consensus protocol in distributed ML is byzantine fault tolerance~\cite{yin2018byzantine, yang2019bridge}, where a byzantine-resilient function filters all possible malicious data before data averaging. However, common byzantine-resilient functions in ML~\cite{yang2019bridge, ghosh2019robust} filter data based on their numerical values (\textit{i.e.,} the largest and the smallest data are filtered) or distances (\textit{i.e.,} the data that have the largest Euclidean distance with the reference is filtered), this inevitably requires the number of byzantine workers (\textit{i.e.,} faulty/malicious worker) in a worker's peers cannot exceed a threshold. Moreover, Byzantine fault tolerance achieves agreements across all workers in the network~\cite{castro1999practical}, and therefore the consensus can only be maintained if the total number of byzantine workers is smaller than 66\%~\cite{sayeed2019assessing}.

\textit{Is it really necessary to maintain a global consensus in the decentralized FL system?} We rethink and conclude that in FL, workers' objectives can be independent of each other. That is, for any worker who participated in FL, its goal can be simply interpreted as enhancing the generalization of its own local model. In that way, each worker can actually just aggregate data from peers they trust, maintaining their own consensus. Consequently, our proposal achieves high model performance where 66\% of workers are malicious\footnote{Due to resource limitation, the number of malicious actors can only be set to a maximum of 66\% of the number of workers in the system.}, without any requirements to the quality of peers.

\begin{table*}[htbp]
  \centering
  \captionsetup{justification=centering}
  \caption{\\\textsc{The Comparison of representational FL frameworks.}}
  \resizebox{\linewidth}{!}{%
    \begin{tabular}{c|ccc|cccc|c}
    \toprule
    \multirow{2}[4]{*}{Training Framework} & \multicolumn{3}{c|}{Security~(~\xmark~is good)} & \multicolumn{4}{c|}{Performance~(~\cmark~is good)} & Others \\
\cmidrule{2-8}          & Servers Failure~\cite{bonawitz2019towards} & Poisoning~\cite{bagdasaryan2020backdoor} & Leakage~\cite{zhu2020deep} & Communication~\cite{konevcny2016federated} & Computation & Model ACC & Straggler~\cite{chen2016revisiting} & (~\xmark~is good) \\
    \midrule
    Single Machine &    -  &    -  &    -   &  -    & \cmark  & \cmark  &    -  & \xmark \\
    \midrule
    On-Site~\cite{mathew2021using} &   -   &    -  &    \xmark   &    -   &  \cmark     & \xmark      &  -     &  \xmark \\
    FedAvg~\cite{mcmahan2016federated} & \cmark  & \cmark  & \cmark    & \xmark     & \cmark      & \cmark   & \xmark  & \xmark \\
    \midrule
    \citet{roy2019braintorrent} &  \xmark  & \cmark  & \cmark  &    \xmark   &  \cmark   &  \cmark  & \xmark  & \xmark \\
    \citet{kim2019blockchained} &  \xmark  & \cmark  & \cmark  &    \cmark   &  \xmark   &  \cmark  & \xmark  & \cmark~(Blockchain) \\
    \citet{hu2019decentralized} &  \xmark  & \cmark  & \cmark  &    \cmark   &  \cmark   &  \cmark  & \xmark  & \xmark \\
    \midrule
    \textbf{DeFTA~(Ours)} &  \xmark  & \xmark  & \xmark  & \cmark   &  \cmark   &  \cmark  & \cmark  & \xmark \\
    \bottomrule
    \end{tabular}%
    }
  \label{tab:framework-comparsion}%
\end{table*}%

\subsection{Related Work}
To achieve privacy-preserving ML, Distributed On-site Learning~\cite{mathew2021using} was proposed to allow each worker to separately train its own model that fits its local dataset. However, since workers will not upload their trained models to others after training, models trained by this framework suffer from low generalization, easy over-fitting, and low accuracy~\cite{abdulrahman2020survey}. After the introduction of FL, highly successful methods~\cite{mcmahan2017communication, mcmahan2016federated, li2018federated, zhao2018federated} are quickly considered one of the most promising technical solutions to the isolated data islands problem, and have been applied in industrial fields~\cite{hard2018federated, brisimi2018federated, samarakoon2019distributed}. Such a centralized system takes the mainstream of FL but also raises new concerns~\cite{kairouz2019advances} about the entire client-server cluster, \textit{e.g.,} network overhead~\cite{konevcny2016federated}, single-point failure of the central server~\cite{kairouz2019advances}, trust issues~\cite{zhu2020deep, zhao2020idlg}, etc. Motivated by decentralized systems in ML~\cite{tsitsiklis1984problems, colin2016gossip}, some decentralized FL studies partially addressed the above-mentioned problems in a peer-to-peer manner~\cite{saroiu2003measuring} by making strong assumptions about either network topology or workers. For instance, recent studies~\cite{roy2019braintorrent, hu2019decentralized} extend the model aggregating algorithm from centralized FL to decentralized FL, assuming the topology is ideal and thereby ignoring the broadcasting procedure in the decentralized system. Moreover, blockchain alike method~\cite{kim2019blockchained} requires workers to store the whole history of models weights locally, and \cite{roy2019braintorrent} forces all workers to be densely connected, making them unpractical for real-world usage.

For illustrative comparison, we summarize 7 aspects for FL framework to consider, which falls into three categories: \textit{Security}, \textit{Performance} and \textit{Others} (Table~\ref{tab:framework-comparsion}). Specifically, \textit{Security} measures whether the given framework is vulnerable to specific attacks, \textit{Performance} measures whether the given framework is capable of efficiently training models, and \textit{Others} shows whether the given framework has other weaknesses. On the other side, in the provided training frameworks for comparison, \textit{Single Machine}, as opposed to distributed FL, means to train models locally, and is served as an ideal reference as it owns high performance without any security issues. Moreover, \textit{On-Site}~\cite{mathew2021using} and \textit{FedAvg}~\cite{mcmahan2016federated} are considered classical centralized privacy-preserving frameworks that are widely used nowadays. Furthermore, \cite{roy2019braintorrent, kim2019blockchained, hu2019decentralized} are considered three representational decentralized FL frameworks. Finally, DeFTA represents our proposal in this paper.

As it can be seen, differ from mentioned decentralized FL frameworks, we consider a more realistic scenario with unreliable peers, privacy-preserving communications, arbitrary network topology, and asynchrony. In other words, In a more realistic scenario, our proposed DeFTA performs comparably with centralized FL and achieves more robust results under malicious environments, while previous decentralized FL methods lose effectiveness.

\begin{figure}[htb]
	\centering
	\includegraphics[width=0.6\linewidth]{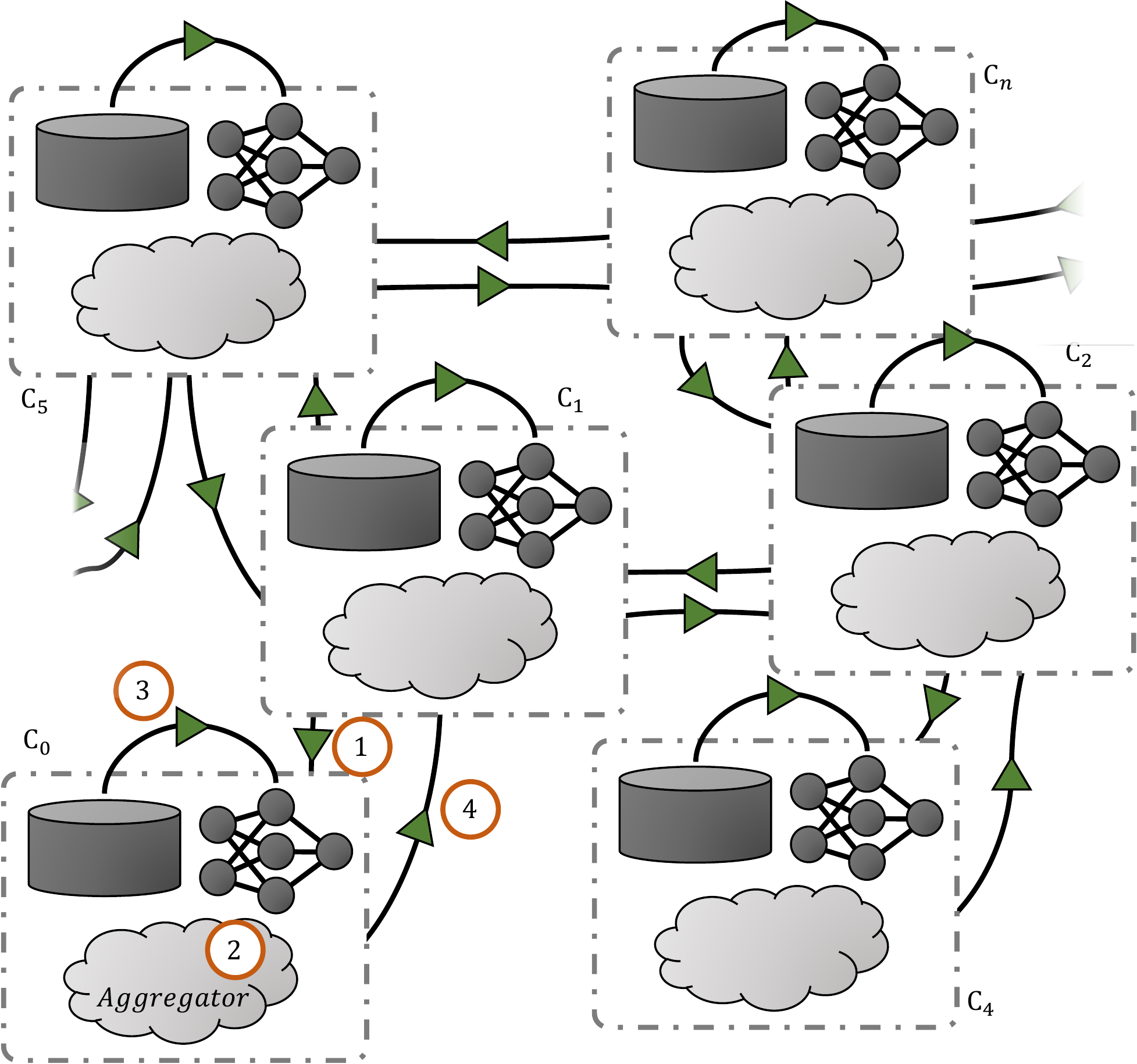}
	\caption{The workflow of our proposed DeFTA, where each of the workers is composed of a private local dataset, a local model, and a local aggregator. Workflow: 1. each worker samples a subset of peers to receive their models, 2. each worker aggregates peers' models using out-degrees related weights, and produces a new local model 3. each worker trains its new local model using its local dataset. 4. each worker sends its trained local model to all its peers.}
	\label{fig:defta-workflow}
\end{figure}

\section{Methodology}
\label{sec:methodology}

\subsection{Overview}

In DeFTA, workers are connected in a P2P manner as a graph, where vertices denote workers and edges denote P2P connections. DeFTA features three primary properties:

\textbf{Performance}: Figure~\ref{fig:defta-workflow} shows each worker in DeFTA will only directly communicate with its $1$-hop peers, significantly reducing the communication load during training. Moreover, out-degrees related weights are used in model aggregation to produce unbiased global model compared to centralized FL (\textit{i.e.,}~\cite{mcmahan2017communication}).

\textbf{Trustlessness}: Real-world FL applications are often vulnerable to malicious actors in the training process~\cite{kairouz2019advances}. In consequence, properly identifying and eliminating them is seen as one of the most challenging tasks in FL. In DeFTA, thanks to the sparse connection brought by decentralized architecture, instead of establishing a global consensus(\textit{i.e.,} byzantine fault tolerance~\cite{yin2018byzantine, yang2019bridge}), DTS is proposed to tackle this problem by assuming workers are all selfish, and will only receive models that could benefit themselves. In other words, each worker can be seen as a gateway in the data paths from one of its peers to other workers, with the power to shut down these data paths. Consequently, malicious actors can be effectively separated from the main network if all their peers drop the connection from them (Figure~\ref{fig:trust-system-overview}), hence achieving the consensus.

\begin{figure}[htp]
	\centering
	\includegraphics[width=0.7\linewidth]{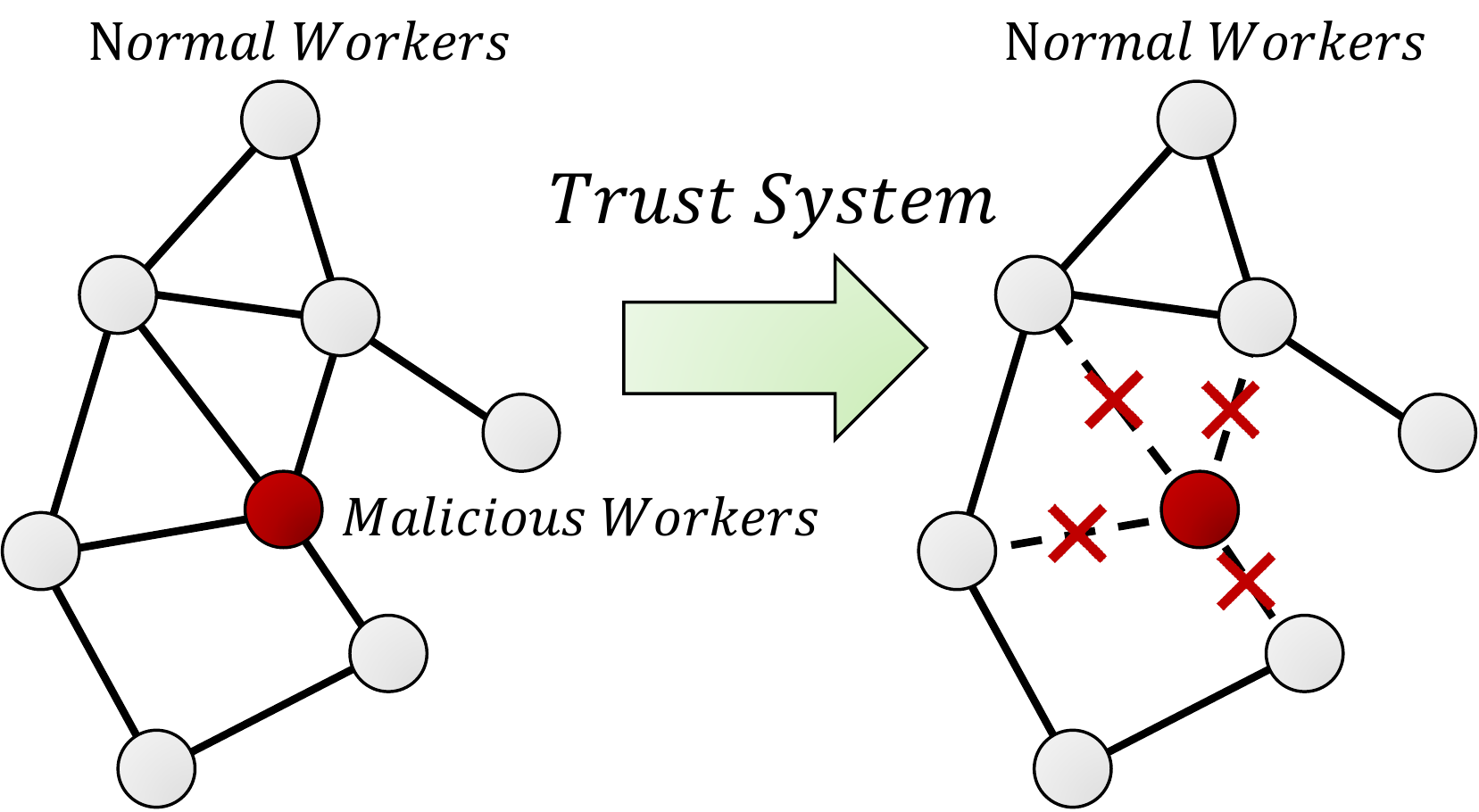}
	\caption{Each worker in the network is constantly evaluating all its peers' sent models, and will shut down connections from peers that they think are faulty. Consequently, DTS in DeFTA can automatically detect malicious actors and separate them from the main network.}
	\label{fig:trust-system-overview}
\end{figure}

\textbf{Asynchronism}: In a practical decentralized system, each worker is only aware of its peers, natively lacking a global clock for the cluster to synchronize, and therefore contributing additional communication overhead to achieve synchronization. Thus, asynchronism is generally required to maximize the utilization and performance of decentralized FL. In our case, we achieve asynchronism by maintaining a synchronized sub-FL-system asynchronously, where a sub-FL-system means any worker with its peers. That is, it is still synchronized within each sub-FL-system, but is asynchronous between different sub-FL-systems. Such a designation is efficient, easy to implement, and holds maximum compatibility for previous algorithms, thus requiring minimal efforts for them to migrate to our framework.

\begin{algorithm}[htb] 
	\caption{DeFTA}
	\label{alg:defta-overview} 
	\begin{algorithmic} 
	    \Definition
            \State $Connect_{n_j}(i)$: invoking $n_j$'s $OnConnect(\cdot)$.
    	    \State $Send(\cdot)$: invoking $n_j$'s $OnReceive(\cdot)$.
    	    \State $\varphi(\cdot)$: model aggregating function.
    	    \State $\phi(\cdot)$: DTS algorithm.
	    \DeFTA 
	    \For {each worker $n_i$ parallel}
	        \For {each worker $n_j \in \mathcal{N}_i$}
	           \State $|\mathcal{D}_j|, d_j, \tilde{w}_j = Connect_{n_j}(i)$
	        \EndFor
	        \State Initialize local model $w_i^0$
	        \State Initialize $\mathcal{S}_i^0 = \mathcal{N}_i$
	        \State Initialize confidence scores $\mathbf{c}_{i}^0$
	        \For {each global iteration $t \in {1, 2, ...}$}
	            \State Aggregating: $w_i^t = \varphi(\{ |\mathcal{D}_j|, d_j, \tilde{w}_j | n_j \in \mathcal{S}_i^t\})$
	            \State Local optimizing: $\tilde{w}_i^t = w_i^t - \eta_i^t \nabla \mathcal{L}_i(w_i^t)$
	            \State Update confidence scores: $\mathbf{c}_{i}^{t+1}, \mathbf{\theta}_{i}^{t+1} = \phi(\mathbf{c}_{i}^{t}, \tilde{w}_i^t)$
	            \State Sample peers: $\mathcal{S}_i^{t+1} = WeightedSample(\mathcal{W}_i^t, \theta_i^{t})$
	            \State Send local data: $Send(\{|\mathcal{D}_i|, d_i, \tilde{w}_i^t, j | n_j \in \mathcal{S}_i^t\})$
	            \State $WaitUntilAllPeersInEpoch(t, \mathcal{W}_i)$
	        \EndFor
	    \EndFor
	    \Blank
	    \OnConnect
	        \State Send necessary information: $Send(|\mathcal{D}_i|, d_i, \tilde{w}_i^t, j)$
	    \Blank
	    \OnReceive
	        \State Update local information: $|\mathcal{D}_j|, d_j, \tilde{w}_j = |\mathcal{D}_j|, d_j, w_j$
	\end{algorithmic} 
\end{algorithm}

Without loss of generality, we firstly suppose that $\mathcal{N}=\{n_i\}$ denotes the set of all participating workers. Then, $\forall$ worker $n_i$, $\mathcal{N}_i$ is its peers (\textit{i.e.}, its neighbors in the graph), $\mathcal{D}_i$ is its local dataset, $d_i$ is its outdegree in the network topology, $\mathcal{L}_i(\cdot)$ is its local loss function, and $\eta_i$ is its learning rate at epoch $t$. Meanwhile, for any connection between $n_i$ and $n_j$ at epoch $t$, $c_{i \rightarrow j}^t$ and $\theta_{i \rightarrow j}^t$ are the confidence score and sample weight from $n_i$ to $n_j$ respectively. Finally, $\mathcal{S}_i^t$ is $n_i$'s sampled peers using $\theta_{i \rightarrow j}^t$, which will be explained in Section~\ref{subsec:trust-system}. Hence, the pseudo-code of our proposed DeFTA is illustrated in Algorithm~\ref{alg:defta-overview}, where $\mathbf{c}_{i}^{t} = [c_{i \rightarrow j}^{t}]^{\top} \in \mathbb{R}^{|\mathcal{N}_i|}$, $\mathbf{\theta}_{i}^{t} = [\theta_{i \rightarrow j}^{t}]^{\top} \in \mathbb{R}^{|\mathcal{N}_i|}$ for convenience and $j \in \mathcal{N}_i$.

\subsection{Model Aggregating Formula}
In FL, each worker optimizes its own model by default, and occasionally receives other workers' models to establish a global view, and thereby preventing over-fitting. Specifically, DeFTA seeks to optimize the following objective:
\begin{equation}
    \mathbf{min}_{w} \sum_{i}^{\mathcal{N}} \sum_{j}^{\mathcal{N}_i} p_{i,j} \mathcal{L}_j(w_j)
    \label{eq:dfl-main-problem}
\end{equation}
Here, $p_{i,j}$ equals to $\frac{\frac{|\mathcal{D}_j|}{d_j}}{\sum_k^{\mathcal{N}_i} \frac{|\mathcal{D}_k|}{d_k}}$ instead of $\frac{|\mathcal{D}_j|}{\sum_k^{\mathcal{N}_i}|\mathcal{D}_k|}$, where the latter one repeated results in repeated aggregation of workers' local models from a global perspective, and is commonly used in previous decentralized FL works. Compared to it, we adopt a normalized factor that simply divides each worker's dataset size with its outdegree for revision. To formally validate the intuition, some mild assumptions need to be made as follows.
\begin{assumption}
\label{assump:mild-assumption}
    For any worker $i$ in the network, its local dataset size and outdegree are strictly larger than $0$ and independent from other workers, and obey $|\mathcal{D}_i| \sim \binom{n_\mathcal{D}}{p_\mathcal{D}}$ and $d_i \sim \binom{n_d}{p_d}$.
\end{assumption}
Note that workers' outdegrees can be independent because connections are directional.

\begin{assumption}
\label{assump:omega}
	Suppose $\Omega^t \in \mathbb{R}^{|\mathcal{N}| \times |\mathcal{N}|}$ is a square matrix whose entry $\omega_{i,j}^t$ intuitively represents the proportion of $w_j^0$ to $w_i^t$. Initially, $\Omega^0 = E$.
\end{assumption}
Consequently Algorithm~\ref{alg:defta-overview} can be casted into a Markov Transition Matrix for transforming $\Omega^{t-1}$ into $\Omega^t$, and thus different $\Omega^t$ can be considered as different states in a Markov Decision Process (MDP), formulating as:
\begin{equation}
     \Omega^{t+1} = P \Omega^{t}\text{,}
     \label{eq:mdp-formula}
\end{equation}
where $P \in \mathbb{R}^{|\mathcal{N}| \times |\mathcal{N}|}$ is a irreducible and ergodic primitive stochastic matrix and is formed by:
\begin{equation}
P_{ij}^t=\left\{
\begin{aligned}
    p_{i,j} & ~,~j \in \mathcal{N}_i \\
    0 ~~~~~ & ~,~j \notin \mathcal{N}_i \\
\end{aligned}
\right.\text{,}
\end{equation}
Then, we can extend Assumption~\ref{assump:mild-assumption} and Assumption~\ref{assump:omega} to Lemma~\ref{lemma:mild} and Lemma~\ref{lemma:1}.
\begin{lemma}
\label{lemma:mild}
    Since both $|\mathcal{D}_i|$ and $d_i$ are strictly larger than $0$, we can make estimations thanks to Taylor expansion for the moments of functions of random variables.
    \begin{enumerate}
        \item $\sum_j^{\mathcal{N}_i}|\mathcal{D}_j| \sim \binom{d_i n_{\mathcal{D}}}{p_{\mathcal{D}}}$.
        \item $\mathbb{E}\frac{1}{|\mathcal{D}_i|} \approx \frac{1}{n_{\mathcal{D}}p_{\mathcal{D}}}$, $\mathbb{E}\frac{1}{d_i} \approx \frac{1}{n_d p_d}$, and $\mathbb{E}\frac{|\mathcal{D_i}|}{d_i} \approx \frac{n_{\mathcal{D}}p_{\mathcal{D}}}{n_d p_d}$.
    \end{enumerate}
\end{lemma}
\begin{lemma}
\label{lemma:1}
	Since $P$ is irreducible and ergodic, by Ergodic Theorem, powers of $P$ converges to a stationary value $\pi$. Formally, $\lim_{n \rightarrow +\infty} (P)^n = \pi$, where $\pi = (\pi_0, \pi_1, \dots, \pi_{|\mathcal{N}|-1})\in \mathbb{R}^{|\mathcal{N}| \times |\mathcal{N}|}$ and $rank_\pi = 1$. For simplicity, we let $\pi_i = C$ denotes $\pi_i = C\mathbf{1} \in \mathbb{R}^{|\mathcal{N}| \times 1}$ where $C$ is a constant scalar.
\end{lemma}
Thus, Theorem~\ref{theorem:main-theorem} can be derived.
\begin{theorem} \label{theorem:main-theorem}
    Let Assumption~\ref{assump:omega} holds, with Lemma~\ref{lemma:1}, Aggregations in Equation~\ref{eq:dfl-main-problem} can be normalized as long as $\sum_{i}^{\mathcal{N}_j} \mathbb{E} \frac{|\mathcal{D}_i|}{|\mathcal{D}_j|} p_{i,j}=1,~\forall j \in \mathcal{N}$, 
\end{theorem}
\begin{proof}
This theorem will be proven by reducing DeFTA's model aggregating formula to Fedavg~\cite{mcmahan2016federated}. In other words, we will prove $\lim_{t \rightarrow +\infty}\Omega^t = (\frac{|\mathcal{D}_0|}{|\mathcal{D}|}, \frac{|\mathcal{D}_1|}{|\mathcal{D}|}, \dots, \frac{|\mathcal{D}_{|\mathcal{N}|-1}|}{|\mathcal{D}|})$, where $|\mathcal{D}|=\sum_i^{\mathcal{N}}|\mathcal{D}_i|$

Noticed that Equation~\ref{eq:mdp-formula} can be expanded to Equation~\ref{eq:mdp-formula-3}.

\begin{equation}
    lim_{t \rightarrow +\infty}\Omega^{t+1} = (P)^t \Omega^0
\label{eq:mdp-formula-3}
\end{equation}

By Lemma~\ref{lemma:1}, it holds $\pi P = \pi$. Let $\pi = (\frac{|\mathcal{D}_0|}{|\mathcal{D}|}, \frac{|\mathcal{D}_1|}{|\mathcal{D}|}, \dots, \frac{|\mathcal{D}_{|\mathcal{N}|-1}|}{|\mathcal{D}|})$, we have Equation~\ref{eq:p-convergence}.

\begin{equation}
\begin{aligned}
    \pi_j &= \frac{|\mathcal{D}_j|}{|\mathcal{D}|} = \sum_{i}^{\mathcal{N}} \pi_i p_{i,j} = \sum_{i}^{\mathcal{N}_j} \pi_i p_{i,j} = \sum_{i}^{\mathcal{N}_j} \frac{|\mathcal{D}_i|}{|\mathcal{D}|} p_{i,j} \\
    &= \sum_{i}^{\mathcal{N}_j} \frac{|\mathcal{D}_j|}{|\mathcal{D}|} \frac{|\mathcal{D}_i|}{|\mathcal{D}_j|} p_{i,j} = \frac{|\mathcal{D}_j|}{|\mathcal{D}|} \sum_{i}^{\mathcal{N}_j} \frac{|\mathcal{D}_i|}{|\mathcal{D}_j|} p_{i,j}
\end{aligned}
\label{eq:p-convergence}
\end{equation}

Consequently, in order to satisfy Equation~\ref{eq:p-convergence}, $\sum_{i}^{\mathcal{N}_j} \frac{|\mathcal{D}_i|}{|\mathcal{D}_j|} p_{i,j}$ must be equal to $1$. Hence the proof is complete.
\end{proof}
\begin{corollary}
\label{corollary:1}
By Lemma~\ref{lemma:mild} and Theorem~\ref{theorem:main-theorem}, model aggregations will be biased if $p_{i,j} = \frac{|\mathcal{D}_j|}{\sum_k^{\mathcal{N}_i}|\mathcal{D}_k|}$.
\end{corollary}
\begin{proof}
$\sum_i^{\mathcal{N}_j} \mathbb{E} \frac{|\mathcal{D}_i|}{|\mathcal{D}_j|} p_{i,j}$ can be firstly reduced to Equation~\ref{eq:calofest-1}.
\begin{equation}
\begin{aligned}
    \sum_i^{\mathcal{N}_j} \mathbb{E} \frac{|\mathcal{D}_i|}{|\mathcal{D}_j|} p_{i,j} &= d_j \mathbb{E} \frac{|\mathcal{D}_i|}{|\mathcal{D}_j|} p_{i,j} = d_j \mathbb{E} \frac{|\mathcal{D}_i|}{\sum_{k}^{\mathcal{N}_i} |\mathcal{D}_k|} \\
    &= d_j \mathbb{E} |\mathcal{D}_i| \mathbb{E} \frac{1}{\sum_{k}^{\mathcal{N}_i} |\mathcal{D}_k|}
\end{aligned}
\label{eq:calofest-1}
\end{equation}
Based on Lemma~\ref{lemma:mild}, we have biased estimation:
\begin{equation}
    \sum_i^{\mathcal{N}_j} \mathbb{E} \frac{|\mathcal{D}_i|}{|\mathcal{D}_j|} p_{i,j} \approx \frac{d_j}{d_i} \neq 1
\end{equation}
\end{proof}
Here in Corollary~\ref{corollary:1}, we show that by utilizing the ratio of workers' local dataset sizes, model aggregations in decentralized FL will be biased (\textit{i.e.,} unequal to $1$) compared to centralized FL, which is mainly caused by variant $d$ under decentralized settings. To rectify it, we additionally normalize $p_{i,j}$ by $d$, which is illustrated in Corollary~\ref{corollary:2}.
\begin{corollary}
\label{corollary:2}
By Lemma~\ref{lemma:mild} and Theorem~\ref{theorem:main-theorem}, model aggregations will be unbiased if $p_{i,j} = \frac{\frac{|\mathcal{D}_j|}{d_j}}{\sum_k^{\mathcal{N}_i}\frac{|\mathcal{D}_k|}{d_k}}$.
\end{corollary}
\begin{proof}
$\sum_i^{\mathcal{N}_j} \mathbb{E} \frac{|\mathcal{D}_i|}{|\mathcal{D}_j|} p_{i,j}$ can be reduced to Equation~\ref{eq:calofest-3}.
\begin{equation}
    \sum_i^{\mathcal{N}_j} \mathbb{E} \frac{|\mathcal{D}_i|}{|\mathcal{D}_j|} p_{i,j} = \sum_{i}^{\mathcal{N}_j} \mathbb{E} \frac{\frac{|\mathcal{D}_i|}{d_j}}{\sum_{k}^{\mathcal{N}_i} \frac{|\mathcal{D}_k|}{d_k}} = \mathbb{E} \frac{\frac{|\mathcal{D}_i|}{d_i}}{\frac{|\mathcal{D}_k|}{d_k}}
\label{eq:calofest-3}
\end{equation}
Based on Lemma~\ref{lemma:mild}, we have unbiased estimation:
\begin{equation}
    \sum_i^{\mathcal{N}_j} \mathbb{E} \frac{|\mathcal{D}_i|}{|\mathcal{D}_j|} p_{i,j} \approx 1
\end{equation}
\end{proof}

Additionally, $\mathcal{N}_i$ in $p_{i,j}$ can be replaced by a sampled subset of $\mathcal{N}_i$, namely $\mathcal{S}_i$, without influence the estimation. Hence, $\varphi(\cdot)$ in Algorithm~\ref{alg:defta-overview} can be further expressed in Algorithm~\ref{alg:phi-function}.
\begin{algorithm}[htb] 
	\caption{$\varphi(\{ |\mathcal{D}_j|, d_j, \tilde{w}_j | j \in \mathcal{S}_i^t\})$}
	\label{alg:phi-function} 
	\begin{algorithmic}
	    \funcvarphi
    	    \State Model aggregating: $w_i^t = \sum_j^{\mathcal{S}_i^t} p_{i,j}^{t} \tilde{w}_j$
    	    \State \textbf{Return} $w_i^t$
	\end{algorithmic} 
\end{algorithm}

Finally, there are already works~\cite{li2019convergence} that comprehensively studied the convergence rate of centralized FL. Hence, the convergence rate of DeFTA will be given by reduction. Since $\lim_{t \rightarrow +\infty}\Omega^t$ will converge to $\pi$ in DeFTA as Theorem~\ref{theorem:main-theorem} suggested, it can be viewed as workers' local models will gradually approach the global model $w_g^t$ with communications with peers, similar to \cite{koloskova2020unified}.
\begin{assumption}
\label{assump:3}
    At each communication, for the $i$-th worker where $i \in range(0, \mathcal{N})$, there has
    \begin{equation}
        || w_{i}^{t+1} - w_g^t || = \beta^t || \tilde{w}_{i}^{t} - w_g^t  ||\text{,}
    \end{equation}
    where $w_g^t = \sum_{i}^{\mathcal{N}} \frac{|\mathcal{D}_i|}{\sum_{j}^{N} |\mathcal{D}_j|} \tilde{w}_{i}^t$ is equivalent to the global model in FedAvg, $\beta^t$ is a convergence factor only related to $P$.
\end{assumption}
Hence, the update of DeFTA can be described as:
\begin{equation}
\begin{aligned}
	& \tilde{w}_i^{t}=w_i^{t}-\eta_i^t \nabla \mathcal{L}(w_i^t) \\
	& w_i^{t+1}= \left\{
    \begin{aligned}
        \tilde{w}_i^{t}~~~~~~~~~~~~ & ~,~Local~training \\
        (1 - \beta^t) w_g^t + \beta \tilde{w}_i^{t}  & ~,~Communicating \\
    \end{aligned}
    \right.
\end{aligned}
\label{eq:model-optimizing-2}
\end{equation}
Thus, DeFTA's updating process can be reduced to centralized FL, where instead of directly marching to $w_g^t$, workers in DeFTA are only expected to move a partial distance in the direction of $w_g^t$. Consequently, the convergence rate of DeFTA is the same as centralized FL. Note that clearly the introduction of $\beta^t$ unavoidably causes larger variance for $w_i^{t+1}$ compared to centralized FL. However, this variance can be effectively reduced by decreasing either the sparseness of $P$ or the learning rate $\eta$.

\subsection{Decentralized Trust System}
\label{subsec:trust-system}
Existing decentralized FL approaches are designed for ideal networks without malicious agents, raising trust issues. Thus, we propose a novel pessimistic subsystem in the decentralized FL domain, namely \textit{Decentralized Trust System}, for trustless network communication in real-world applications. In DTS, the type of attack is not assumed, instead, it only focuses on the nature of any attack: \textit{model performance reduction}. Thus, DTS considers every worker a selfish learner that will only accept models that help its own model perform better. That is, each worker $w_i$ will continuously evaluate the effectiveness of models sent by its sampled peers $\mathcal{S}_i^t$, and thereby maintaining confidence scores $\mathbf{c}_{i}^t$ for its connections. Moreover, due to the diverse definitions of model performance in different contexts (\textit{e.g.,} lower perplexity in NLP, higher AUC in KT, higher mAP in CV, etc.), the evaluation metrics for models are not pre-defined in DeFTA for higher flexibility of the framework, and can be varied under different scenarios, be it training loss, local accuracy, or others. In this paper, the evaluation metric is the training loss. In other words, the confidence $ c_{i \rightarrow j}^t$ between worker $i$ and $j$ increases if the training loss of $w_i$ decreases, and vise-versa.

Naively taking all peers into the updating of $c_{i \rightarrow j}^t$ hinders the distinguishing of malicious and vanilla workers since their tendencies of confidence updating are identical. Thus, we adopt a sampling procedure to better locate the vulnerability by transforming $c_{i \rightarrow j}^t$ into sample weights $\theta_{i \rightarrow j}^t$. Such a transformation function needs to satisfy the following constraints: (1) discourage low confidence scores, \textit{i.e.,} unreliable peers are not likely to be sampled; (2) encourage long-term commitments, \textit{i.e.,} confidence scores for reliable workers can still be slowly increased; (3) treat vanilla peers equally, \textit{i.e.,} reliable workers roughly share the identical sample probability. Initially, the confidence scores of all workers are initialized to $0$, referring to a neutral status. Formally, by applying $cRELU(\cdot)$ and $softmax(\cdot)$ function to $c_{i \rightarrow j}^t$, the transformed sample weight $\theta_{i \rightarrow j}^t$ can be further obtained as following:
\begin{equation}
\theta_{i}^t = softmax(cRELU(c_{i}^t))\text{,}
\label{eq:sampling-weight}
\end{equation}
where $cRELU(\cdot)$ is described by Equation~\ref{sampling-weight-append}.
\begin{equation}
cRELU(x)=\left\{
\begin{aligned}
    x & ~,~x \leq 0 \\
    0.2 x & ~,~x > 0 \\
\end{aligned}
\right.
\label{sampling-weight-append}
\end{equation}
On the one hand, the peers with negative confidence scores are significantly penalized by assigning larger gradients with $cRELU(\cdot)$ (constraint 1). On the other hand, a nearly flat linear function in $cRELU(\cdot)$ slows down the increasing tendency of positive confidence (constraint 2), making them grow together to roughly remedy the positive feedback, \textit{i.e.,} the sample probability of a worker converges to $1$ fastly. Additionally, we adopt $softmax(\cdot)$ normalizer to roughly equalize the sample probabilities of positive confidences (constraint 3). That is, in a neighborhood that has $k$ malicious workers, DTS is expected to achieve $\mathbb{E}\theta_{i \rightarrow j}^t = \frac{1}{d_i - k}$. Consequently, series of $\theta_{i \rightarrow j}^t$ is convergent with respect to epoch $t$.

\textbf{Time Machine of DTS:} Moreover, sometimes malicious actors broadcast extremely dirty models and even aggregating for $1$ time will result in an un-trainable model, \textit{e.g.,} $+\infty$ or carefully constructed model weights. To mitigate this problem, a backup mechanism is embedded in DTS, where it will automatically backup the latest stable local model defined by the DTS's evaluation metrics. When the model becomes damaged, this mechanism will restore the model from the backup. Notably, the worker will train the local model one more time after recovering for compensation.

In detail, the persuade-code of DTS is illustrated in Algorithm~\ref{alg:gamma-function}, where $p_{i} = (p_{i,j}~|~\forall j \in \mathcal{N}_i)^{\top} \in \mathbb{R}^{|\mathcal{N}_i|}$, $m_{i}^t \in \mathbb{R}^{|\mathcal{N}_i|}$ is a $0$-$1$ matrix, and its $j$-th entry is equal to $1$ if and only if $j \in \mathcal{S}_{i}^t$.

\begin{algorithm}[htb] 
	\caption{$\phi(\mathbf{c}_{i}^t, \tilde{w}_i^t)$}
	\label{alg:gamma-function} 
	\begin{algorithmic}[1] 
	\funcphi
	    \If {$\tilde{w}_i^t$ is damaged}
	        \State Recovery: $\tilde{w}_i^t = w_{backup}$
	        \State Compensation: $\tilde{w}_i^t = \tilde{w}_i^t - \eta_i^t \nabla \mathcal{L}(\tilde{w}_i^t)$
	        \State $loss_{trust} = +\infty$
	    \Else
	        \If {$loss^t$ is lowest}
	            \State Make a backup: $w_{backup} = \tilde{w}_i^t$
	        \EndIf
	        \State $loss_{trust} = loss^t - loss_{last}$
	        \State $loss_{last} = loss^t$
        \EndIf
        \State Update confidence scores: $\mathbf{c}_{i}^{t+1} = \mathbf{c}_{i}^t - m_{i}^t \circ p_{i} loss_{trust}$
        \State Transform to sample weights: $\mathbf{\theta}_{i}^{t+1} = softmax(cRELU(\mathbf{c}_{i}^t))$
        \State \textbf{Return} $\mathbf{c}_{i}^{t+1}$, $\mathbf{\theta}_{i}^{t+1}$
	\end{algorithmic} 
\end{algorithm}

\subsection{Asynchrony}
Existing decentralized FL methods almost ignore the asynchronism in real-world applications, limiting scalability seriously. On the other hand, in conventional distributed model training, asynchronous runtime is seriously hindered by the stale gradients problem, leading to the global model degradation problem~\cite{zhang2015staleness, dutta2018slow}. In decentralized FL, since workers are always exchanging model parameters (which can be decomposed into base model parameters and gradients) instead of pure gradients, the stale gradients problem is actually not applied. That is, each worker only cares about delivering the best model it can train to its peers.

Consequently, AsyncDeFTA is the asynchronous version of DeFTA with all features delivered by DeFTA, and ensures no global model degradation. Specifically, to achieve asynchronism, AsyncDeFTA firstly separates workers into multiple sub-FL-systems, where one sub-FL-system has one central worker and all its peers. Note that this separation is fully decentralized, as every worker is a central worker while being a peer. Also, one worker can be included in different sub-FL-systems, and can be seen as different workers with the same local dataset and local model (\textit{i.e.,} clone). Then AsyncDeFTA allows workers in the same sub-FL-systems to be still synchronized, while making different sub-FL-systems to be fully asynchronous. Here being synchronized means each worker will only upload its model to the central worker in the respective sub-FL-system once between model aggregations. As the cost of maintaining synchronization within a sub-FL-system is light, this strategy not only cheaply achieves asynchronism, but also remain maximum compatibility and flexibility. Algorithm~\ref{alg:defta-overview} illustrating the pseudocode of AsyncDeFTA is simply Algorithm~\ref{alg:defta-overview} without $WaitUntilAllPeersInEpoch(\cdot)$, which further validates the simpleness of the implementation. Similarly, any worker can customize its own training strategies, be it hyper-parameters, peer selecting algorithms, etc.

\begin{figure}[htb]
    \centering
    \subfloat[]{%
        \includegraphics[width=\linewidth]{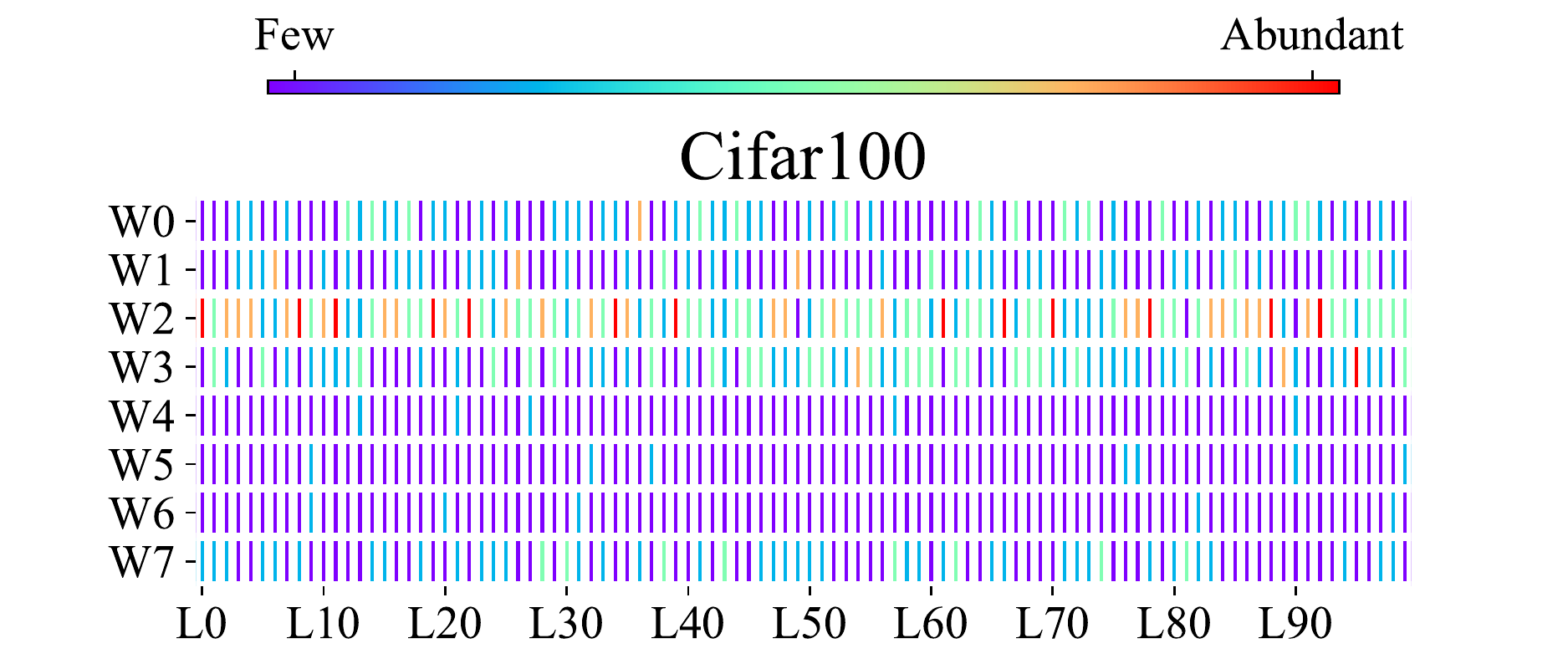}}
    \\
    \subfloat[]{%
       \includegraphics[width=0.45\linewidth]{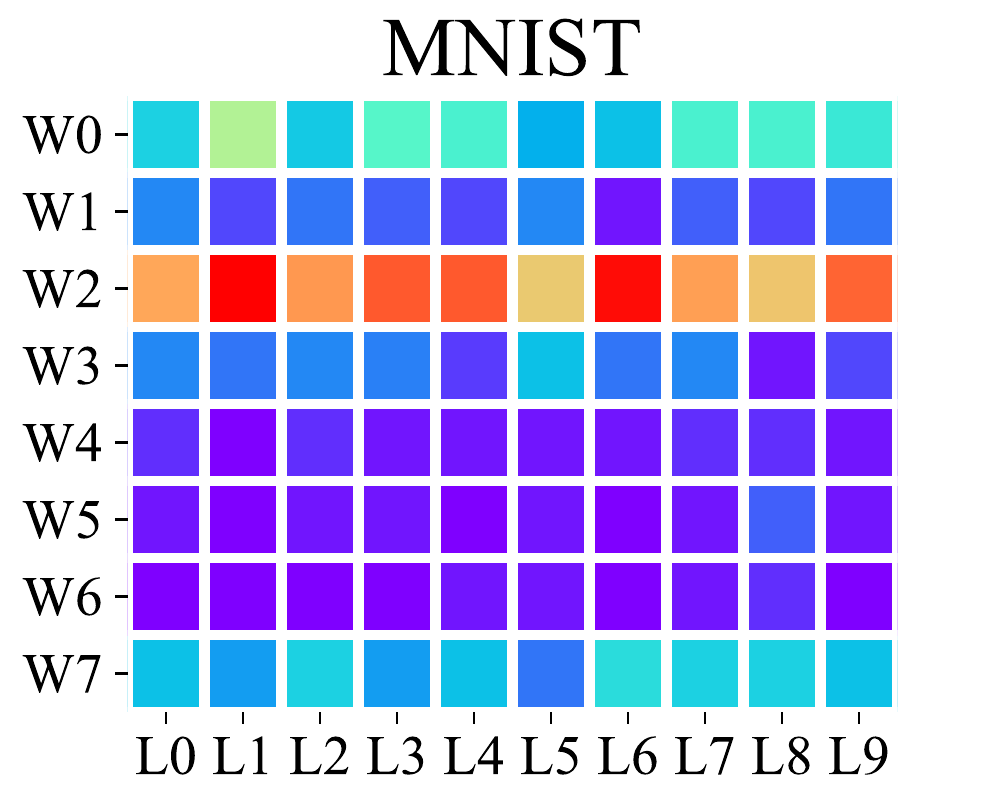}}
    \subfloat[]{%
        \includegraphics[width=0.45\linewidth]{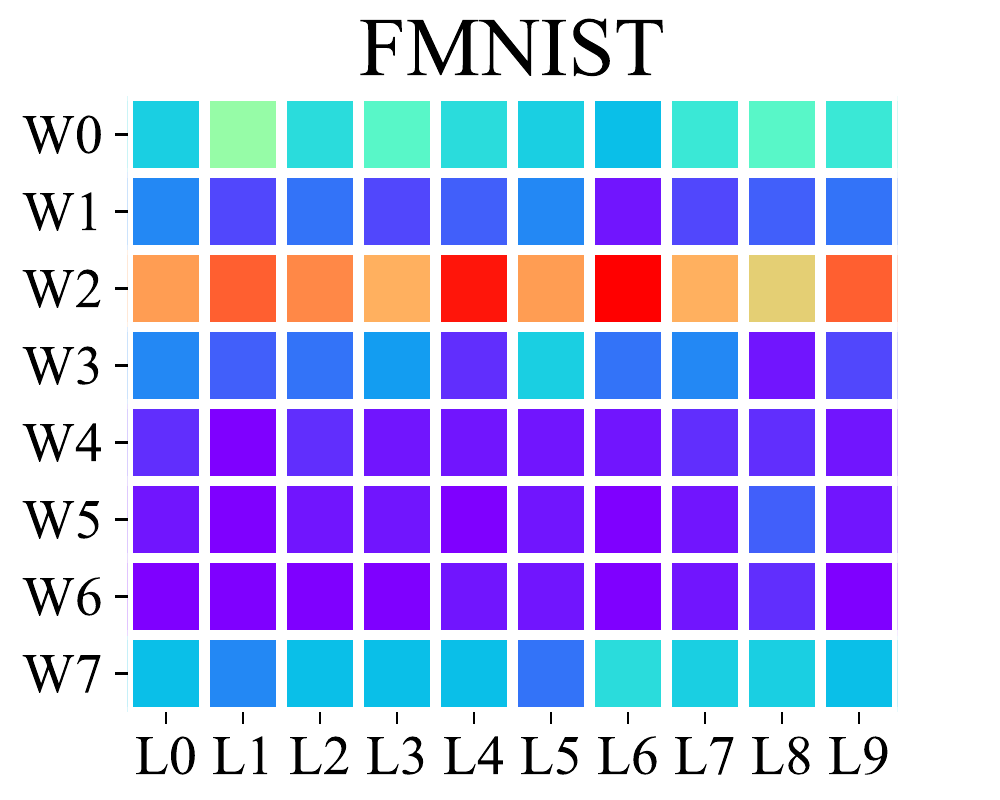}}
    \\
    \subfloat[]{%
        \includegraphics[width=0.45\linewidth]{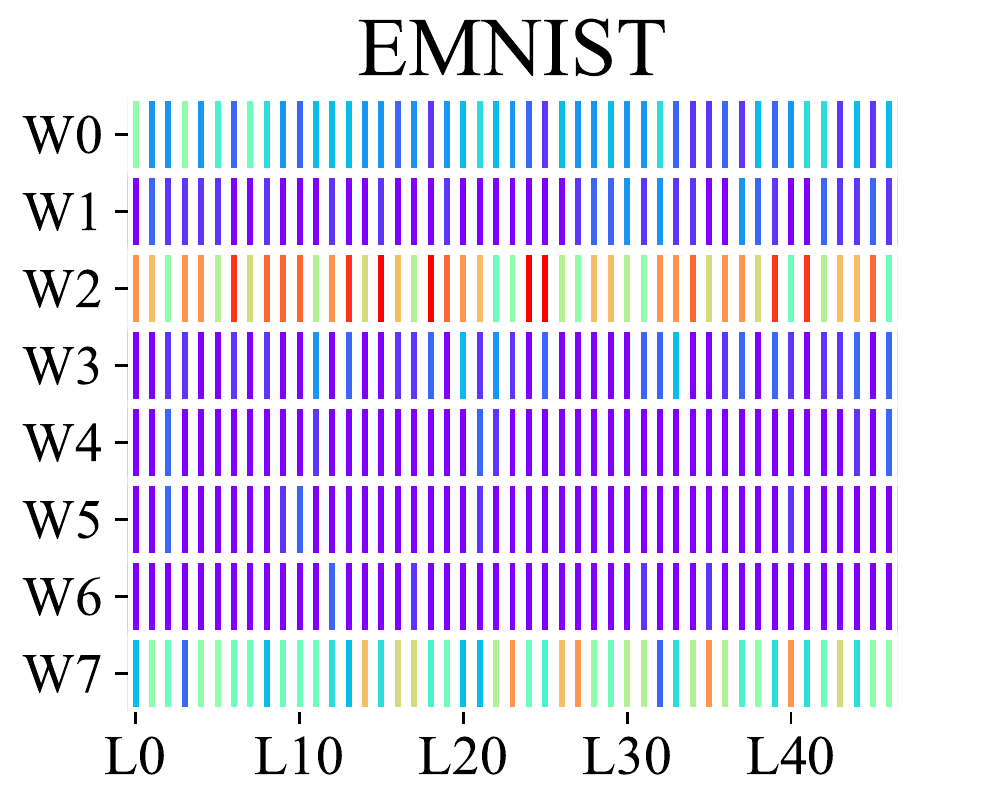}}
    \subfloat[]{%
        \includegraphics[width=0.45\linewidth]{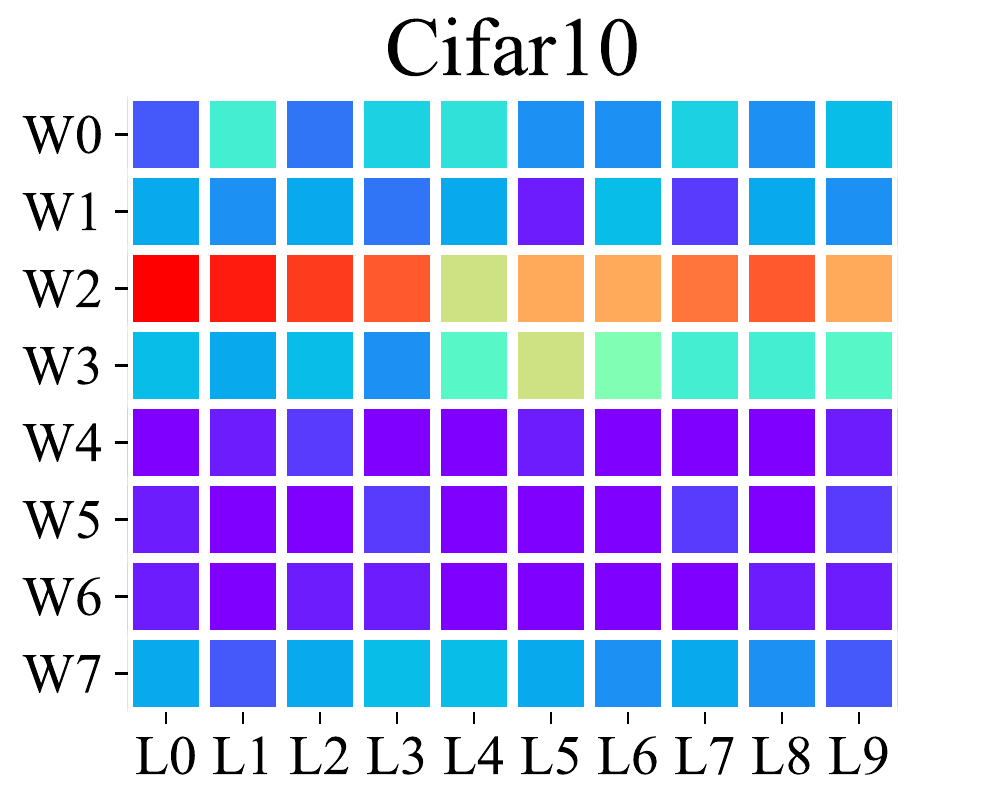}}
    \caption{An illustration of non-i.i.d. Partitioning of multiple datasets for $8$ workers. $W0$ is referred to as worker $0$, $L0$ is referred to as the label of $0$ in the corresponding dataset, and the color of cells is referred to as the number of data it contains.}
    \label{fig:data-distribution} 
\end{figure}

\section{Experiments}
\label{sec:experiments}
\begin{table*}[htbp]
  \centering
  \caption{Accuracy/PPL comparison between centralized FL (CFL) and decentralized FL in $100$ epochs}
  \resizebox{\linewidth}{!}{%
    \begin{tabular}{c|cccc|cccc|cccc}
    \toprule
    \toprule
    \multirow{2}[4]{*}{N, Model, Dataset} & \multicolumn{4}{c|}{8 Workers} & \multicolumn{4}{c|}{14 Workers} & \multicolumn{4}{c}{20 Workers} \\
\cmidrule{2-13}          & CFL-F(\%)$\uparrow$   & CFL-S(\%)$\uparrow$   & DeFTA(\%)$\uparrow$ & DeFL(\%)$\uparrow$  & CFL-F(\%)$\uparrow$   & CFL-S(\%)$\uparrow$   & DeFTA(\%)$\uparrow$ & DeFL(\%)$\uparrow$  & CFL-F(\%)$\uparrow$   & CFL-S(\%)$\uparrow$   & DeFTA(\%)$\uparrow$ & DeFLA(\%)$\uparrow$ \\
    \midrule
    MLP,MNIST & 97.77 & 97.71 & \textbf{97.71$\mathbf{\pm}$0.1} & \textit{97.62$\pm$0.1} & 97.85 & 97.38 & \textbf{97.26$\mathbf{\pm}$0.9} & \textit{96.94$\pm$0.7} & 97.77 & 97.71 & \textbf{97.19$\mathbf{\pm}$0.5} & \textit{96.91$\pm$0.7} \\
    MnistNet,FMNIST & 90.83 & 90.52 & \textbf{90.43$\mathbf{\pm}$0.5} & \textit{89.95$\pm$0.8} & 90.54 & 86.30  & \textbf{88.16$\mathbf{\pm}$1.8} & \textit{87.65$\pm$2.7} & 90.56 & 90.21 & \textbf{88.28$\mathbf{\pm}$1.9} & \textit{87.71$\pm$2.0} \\
    MnistNet,EMNIST & 83.48 & 81.06 & \textbf{81.47$\mathbf{\pm}$2.6} & \textit{80.24$\pm$2.6} & 82.81 & 80.54 & \textbf{80.05$\mathbf{\pm}$2.2} & \textit{79.67$\pm$2.7} & 83.20  & 80.22 & \textbf{79.24$\mathbf{\pm}$2.7} & \textit{78.29$\pm$3.3} \\
    CNNCifar,Cifar10 & 66.75 & 63.36 & \textbf{60.85$\mathbf{\pm}$1.3} & \textit{59.49$\pm$5.2} & 63.35 & 46.61 & \textbf{56.49$\mathbf{\pm}$4.6} & \textit{46.63$\pm$3.1} & 55.54 & 50.56 & \textbf{48.97$\mathbf{\pm}$9.5} & \textit{45.62$\pm$7.2} \\
    VGG,Cifar10 & 80.61 & 78.71 & \textbf{78.14$\mathbf{\pm}$0.9} & \textit{74.55$\pm$4.7} & 79.59 & 76.19 & \textbf{73.23$\mathbf{\pm}$2.6} & \textit{69.71$\pm$4.4} & 78.26 & 73.87 & \textbf{64.94$\mathbf{\pm}$10.4} & \textit{62.95$\pm$12.6} \\
    ResNet,Cifar10 & 80.72 & 78.02 & \textbf{77.61$\mathbf{\pm}$0.9} & \textit{74.64$\pm$0.62} & 78.89 & 76.22 & \textbf{70.52$\mathbf{\pm}$6.0} & \textit{69.68$\pm$3.9} & 77.62 & 73.14 & \textbf{61.74$\mathbf{\pm}$11.2} & \textit{60.51$\pm$10.1} \\
    ResNet,Cifar100 & 41.29 & 37.22 & \textbf{31.3$\mathbf{\pm}$6.0} & \textit{27.68$\pm$5.7} & 36.79 & 28.81 & \textbf{20.19$\mathbf{\pm}$5.7} & \textit{15.80$\pm$4.5} & 36.34 & 23.08 & \textbf{15.72$\mathbf{\pm}$5.4} & \textit{13.52$\pm$6.3} \\
    \midrule
          & CFL-F$\downarrow$   & CFL-S$\downarrow$   & DeFTA$\downarrow$ & DeFL$\downarrow$  & CFL-F$\downarrow$   & CFL-S$\downarrow$   & DeFTA$\downarrow$ & DeFL$\downarrow$  & CFL-F$\downarrow$   & CFL-S$\downarrow$   & DeFTA$\downarrow$ & DeFL$\downarrow$ \\
    \midrule
    Transformer,Wikitext-2 & 1.182 & 1.185 & \textbf{1.193$\mathbf{\pm}$0.001} & \textit{1.194$\pm$0.001}      & 1.188 & 1.193 & \textbf{1.199$\mathbf{\pm}$0.001} & \textit{1.202$\pm$0.001}      & 1.19  & 1.196 & \textbf{1.201$\mathbf{\pm}$0.001} & \textit{1.203$\pm$0.001} \\
    \bottomrule
    \bottomrule
    \end{tabular}%
  }
  \label{tab:model-performance}%
\end{table*}%
\subsection{Experimental Overview}
For verification, we adopt six popular basic models in both CV and NLP and six datasets in our following experiments. In conclusion, we firstly make comparisons between different methods across different peers, and show that our proposed DeFTA achieves comparable performance with centralized FL, while other decentralized approaches fail. Second, we testify the robustness of different methods by simulating wildness where at most 66\% of workers are malicious. The result suggests that our proposed method is capable of successfully completing the model training process, while effectively detecting and separating 66\% of malicious actors from the network natively. Finally, dedicated comparisons of DeFTA and asyncDeFTA are provided, implying asyncDeFTA is as performant as DeFTA without synchronization. More experimental details will be elaborated on below.

\textbf{Models:} There were in total six models used in our experiments, which are listed in ascending order of model size: Multi-Layer Perceptron (MLP), MnistNet, CNNCifar, VGG$^R$~\cite{simonyan2014very}, ResNet$^R$~\cite{he2016deep} and Transformer~\cite{vaswani2017attention}. Note that models with a superscript $R$ represents a reduction of Dropout layers~\cite{srivastava2014dropout} and BatchNorm layers~\cite{ioffe2015batch}. This reduction is trivial and was also used by other works in FL domain~\cite{sattler2019robust}. The detailed specification of all models can be inspected in our open-sourced repository.

\textbf{Datasets:} There were in total six datasets used in our experiments, which are listed in ascending order of dataset difficulty as well: MNIST~\cite{lecun1998mnist}, FMNIST~\cite{cohen2017emnist}, EMNIST~\cite{cohen2017emnist}, Cifar10~\cite{krizhevsky2009learning}, Cifar100~\cite{krizhevsky2009learning} and Wikitext-2~\cite{merity2016pointer}. To match the non i.i.d. characteristics of FL, each dataset was intentionally partitioned into multiple non i.i.d. subsets before training. The demonstrative non i.i.d. partitioning for $8$ different workers is illustrated in Figure~\ref{fig:data-distribution}.

\textbf{Implementation Details:} Comparisons were made in a simulated FL environment where the maximum number of participating workers was $60$ (\textit{i.e.,} world size, including malicious workers). Moreover, the simulated software was coded using PyTorch RPC framework with \texttt{GLOO} communication backend provided by PyTorch 1.8.1. The reason for using \texttt{GLOO} instead of faster backends like \texttt{NCCL} was to ensure the maximum compatibility as \texttt{NCCL} currently only supports Nvidia GPUs, which most edge devices does not have. For hyper-parameters, the global training epoch $E$ was set to $100$, the worker local training epoch was set to $10$, the batch size was set to $64$, and the learning rate was set to $0.01$. These settings were default for all experiments unless otherwise stated. In addition, there is no particular peer selecting strategy applied in DeFTA for generalization, implying that peers of a given worker are randomly selected. However, in practice, the peer selecting strategy can be varied for different downstream tasks, potentially further boosting the performance of DeFTA and asyncDeFTA.

\subsection{Performance Analysis}
We first compare the final model's accuracy/PPL on the test set respectively trained by centralized FL (CFL-A and CFL-S represents FedAvg with all workers and FedAvg with sampled workers respectively, where in the former one the aggregation will involve all workers, and in the latter one it will only involve randomly sampled workers) and Decentralized FL (DeFTA, and DeFTA without proposed model aggregating formula and DTS, which is equivalent to other decentralized FL works~\cite{hu2019decentralized}) with $8$, $14$, and $20$ workers. For CFL-S, the number of sampled workers is set to $2$. For decentralized FL, the average number of peers for any worker is configured to $4$ (\textit{i.e.,} every worker averagely connected to other 4 workers), and the number of sampled workers is set to $2$ (\textit{i.e.,} each worker aggregates $2$ sampled worker's models). The final results of all experiments are shown in Table~\ref{tab:model-performance}.

From Table~\ref{tab:model-performance}, one could find that DeFTA constantly outperforms DeFL across all experimental settings, achieving a comparable model performance compared to CFL-S, and validating the effectiveness of our proposed model aggregating formula. Then, the model performance gap between DeFTA and CFL-S increases concerning the difficulty of the task, which is essentially due to the accumulation of self-training bias to the transformation in MDP, as well as the reduction of connections between workers caused by false-positive eliminations by DTS, as connections between workers are randomly determined, which potentially results in totally unrelated local datasets between connected workers. Fortunately, these issues can be effectively addressed by reducing the learning rate $\eta$ and the sparseness of $P$, or picking other peer selecting strategies and incentives correspondingly. Furthermore, the model performance gap between DeFTA and CFL-S also increases with respect to the world size. In fact, this phenomenon can be observed between CFL-F and CFL-S as well. The reason for such a performance degradation is that partitioned datasets for $20$ workers were much more non-i.i.d. than $8$ workers' (Figure~\ref{fig:data-distribution-ws}), leading to much more difficult tasks.

\begin{figure}[htb]
    \centering
    \subfloat[]{%
       \includegraphics[width=0.4\linewidth]{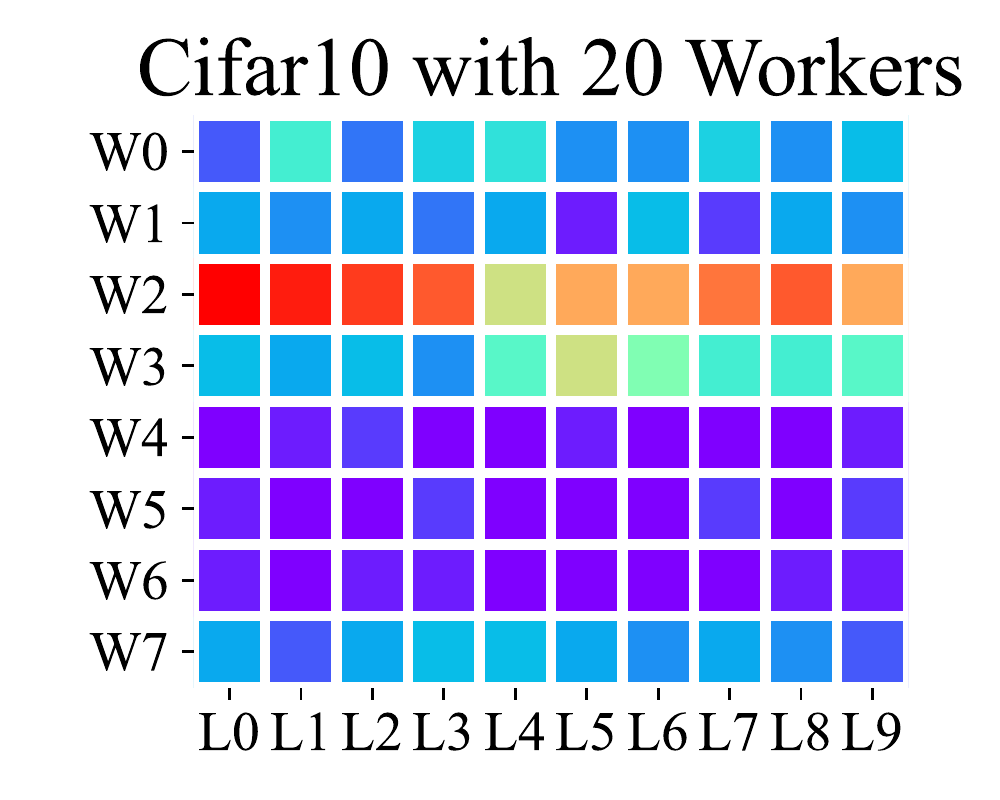}}
    \subfloat[]{%
        \includegraphics[width=0.4\linewidth]{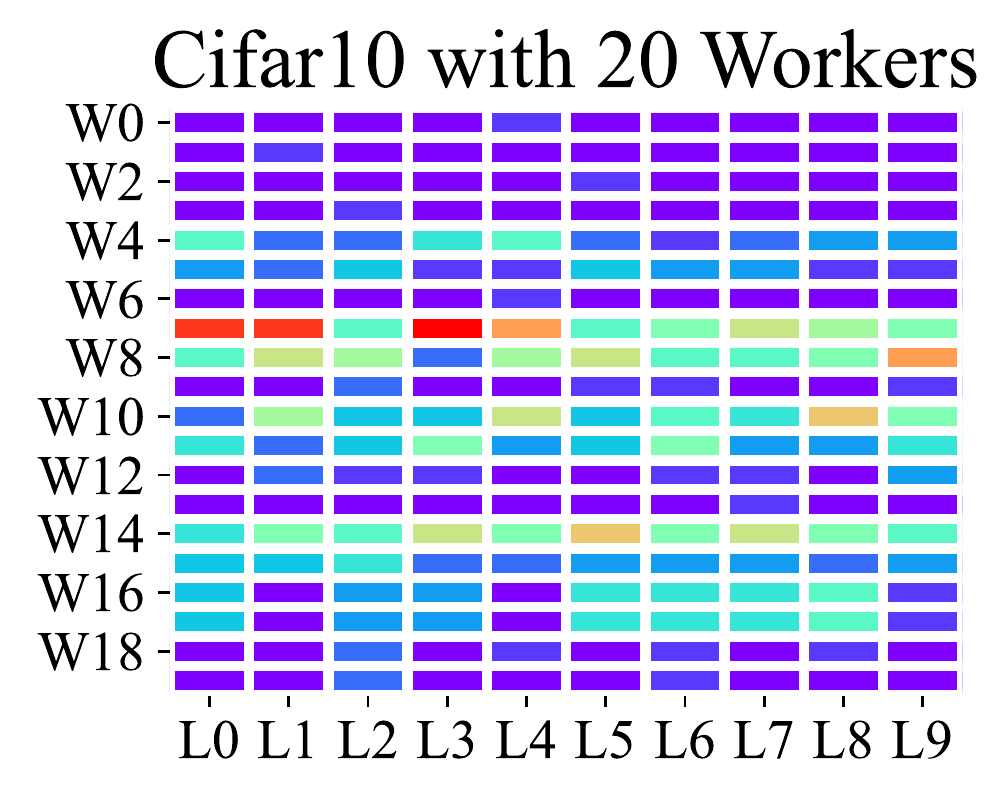}}
    \caption{An illustration of non-i.i.d. partitioning of Cifar10 datasets for $8$ and $20$ workers. It is clear that their data distributions are significantly different, and the latter one is considerably more non-i.i.d. than the former one, which results in a performance degradation with the increment of world size for both FedAvg and DeFTA.}
    \label{fig:data-distribution-ws} 
\end{figure}

\begin{table*}[htbp]
  \centering
  \caption{Accuracy/PPL comparison between methods with $20$ vanilla workers and $k$ malicious actors in $100$ epochs}
  \resizebox{\linewidth}{!}{%
    \begin{tabular}{c|ccc|c|c|c|c|c}
    \toprule
    \multirow{2}[4]{*}{Model,Dataset} & \multicolumn{3}{c|}{k=1(4.8\%)} & k=3(14.3\%) & k=5(20\%) & k=10(33.3\%) & k=20(50\%) & k=40(66.7\%) \\
\cmidrule{2-9}          & CFL-S(\%)$\uparrow$ & DeFL(\%)$\uparrow$  & DeFTA(\%)$\uparrow$ & DeFTA(\%)$\uparrow$ & DeFTA(\%)$\uparrow$ & DeFTA(\%)$\uparrow$ & DeFTA(\%)$\uparrow$ & DeFTA(\%)$\uparrow$ \\
    \midrule
    MLP,MNIST & 11.4  & 10.0$\pm$0.8 & 96.82$\pm$0.5 & 96.94$\pm$0.6 & 96.92$\pm$0.4 & 96.94$\pm$0.4 & 95.01$\pm$2.4 & 90.95$\pm$19.2 \\
    MnistNet,FMNIST & 10.0    & 10.0$\pm$0.0 & 88.16$\pm$1.6 & 87.82$\pm$2.1 & 87.80$\pm$1.9 & 87.78$\pm$2.8  & 88.28$\pm$1.9 & 80.18$\pm$16.9 \\
    MnistNet,EMNIST & 2.1   & 2.1$\pm$0.0 & 78.77$\pm$3.1 & 78.77$\pm$3.1 & 78.04$\pm$3.2 &  76.98$\pm$2.3 & 73.93$\pm$5.3 & 68.36$\pm$16.6 \\
    CNNCifar,Cifar10 & 10.0    & 10.0$\pm$0.0 & 52.90$\pm$8.7 & 45.7$\pm$11.4 & 35.52$\pm$12.0 & 39.51$\pm$15.7 & 37.83$\pm$15.7 & 23.44$\pm$19.3 \\
    VGG,Cifar10 & 10.0    & 10.0$\pm$0.0 & 70.01$\pm$3.7 & 68.06$\pm$6.3 & 60.13$\pm$13.9 & 62.00$\pm$10.1 & 58.65$\pm$5.8 & 43.37$\pm$22.1 \\
    ResNet,Cifar10 & 10.0    & 10.0$\pm$0.0 & 66.57$\pm$6.2 & 67.00$\pm$6.4 & 64.98$\pm$9.1 & 63.11$\pm$10.9  & 64.41$\pm$7.5 & 50.85$\pm$18.3 \\
    ResNet,Cifar100 & 1.0     & 1.0$\pm$0.0 & 17.72$\pm$4.5 & 16.75$\pm$5.8 & 15.70$\pm$5.4 & 12.56$\pm$7.3 & 12.66$\pm$6.4 & 8.56$\pm$6.3 \\
    \midrule
          & CFL-S$\downarrow$ & DeFL$\downarrow$  & DeFTA$\downarrow$ & DeFTA$\downarrow$ & DeFTA$\downarrow$ & DeFTA$\downarrow$ & DeFTA$\downarrow$ & DeFTA$\downarrow$ \\
    \midrule
    Transformer,Wikitext-2 & N/A   & N/A   & 1.201$\pm$0.001 & 1.2$\pm$0.001 & 1.203$\pm$0.001 & 1.198$\pm$0.001 & 1.200$\pm$0.001 & - \\
    \bottomrule
    \end{tabular}%
  }
  \label{tab:dts}%
\end{table*}%

\begin{figure*}[h]
    \centering
    \subfloat{%
        \includegraphics[width=0.25\linewidth]{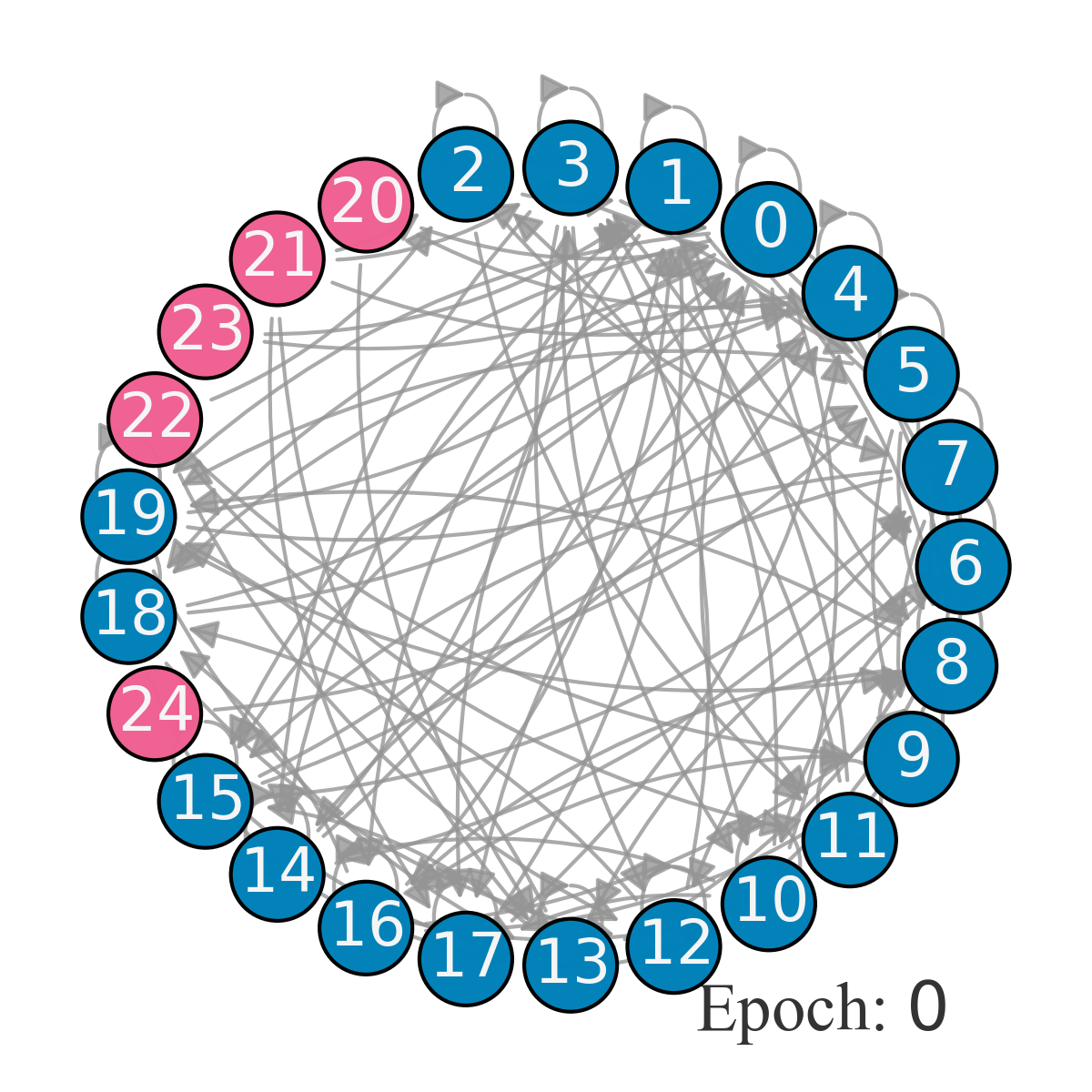}}
    \subfloat{%
        \includegraphics[width=0.25\linewidth]{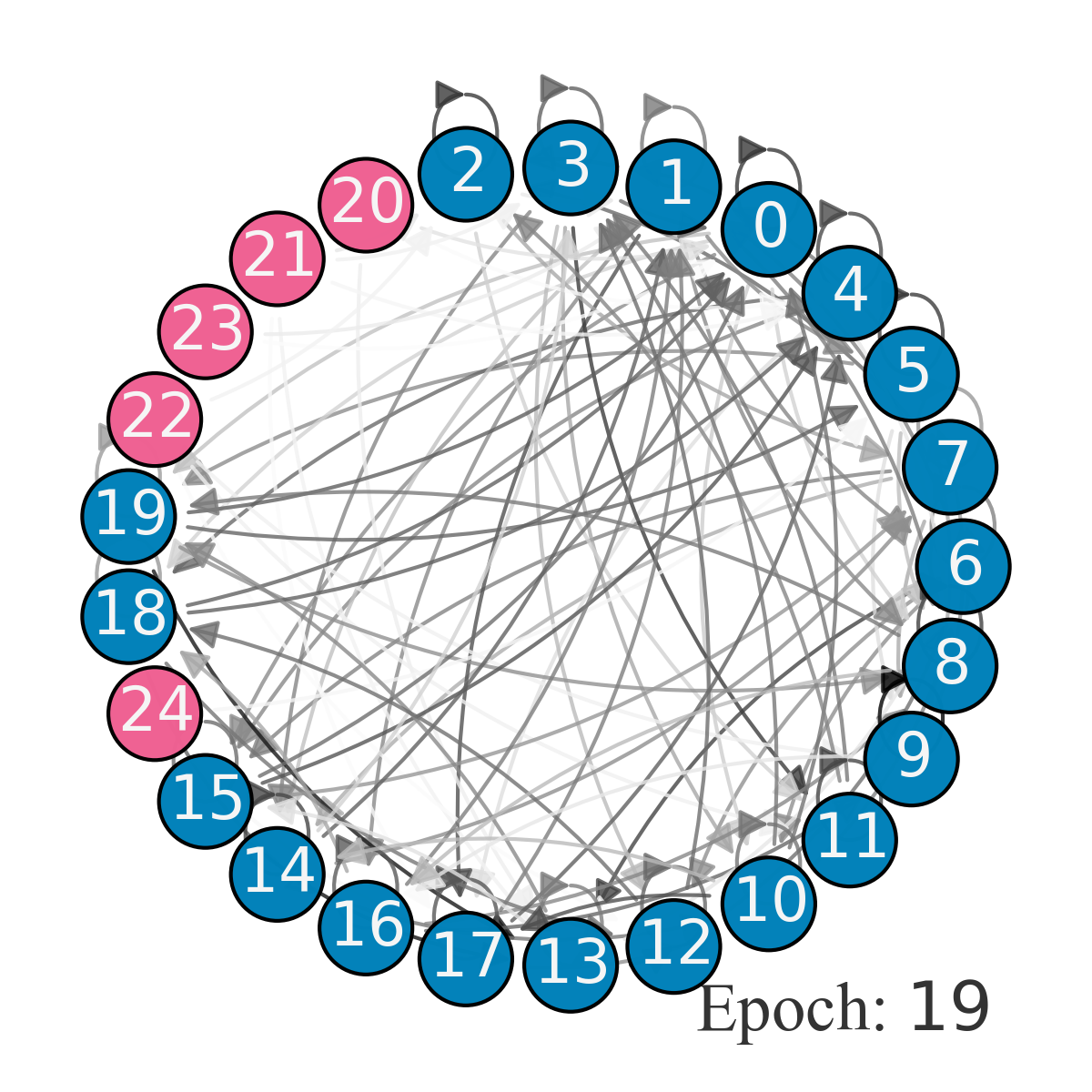}}
    \subfloat{%
        \includegraphics[width=0.25\linewidth]{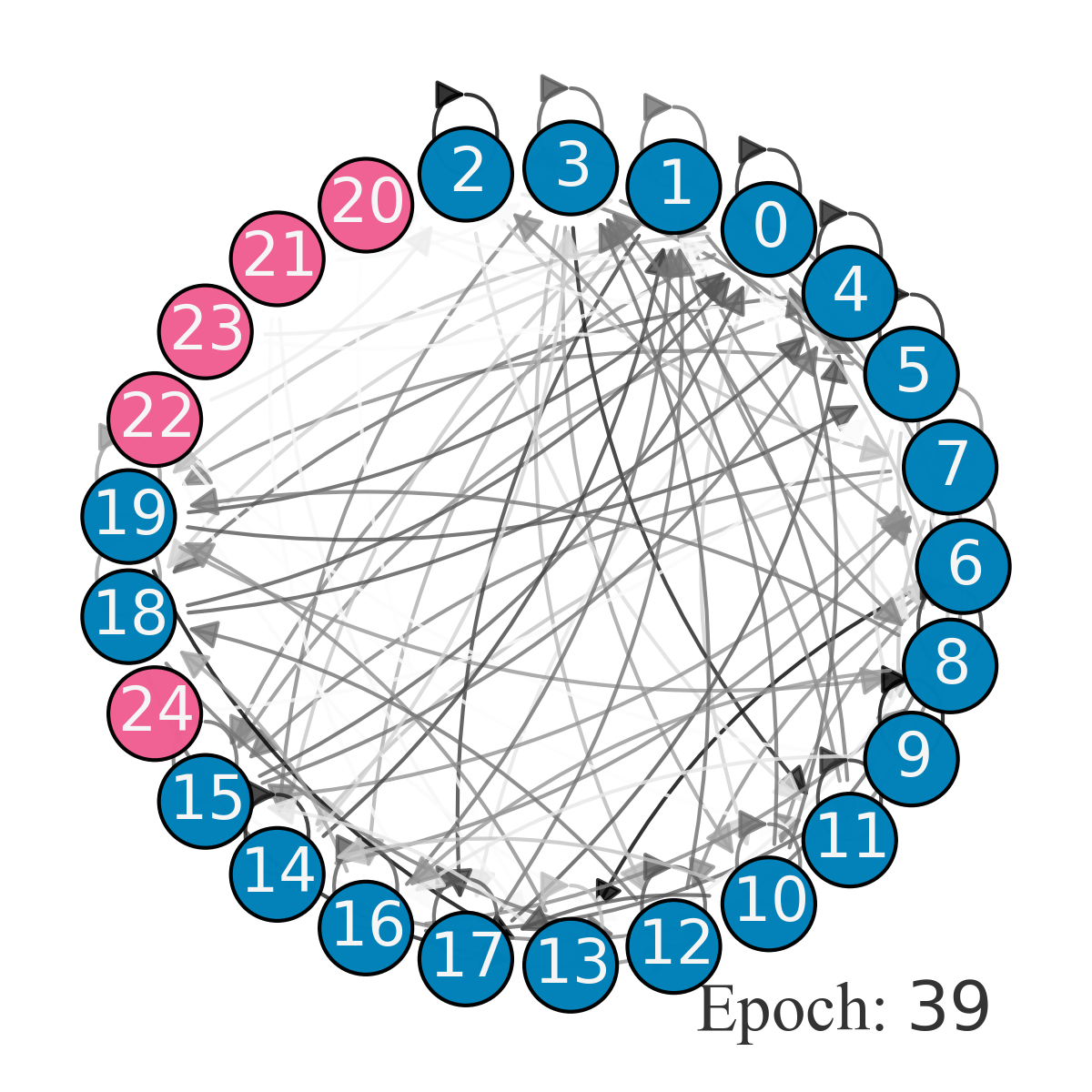}}
    \\
    \subfloat{%
      \includegraphics[width=0.25\linewidth]{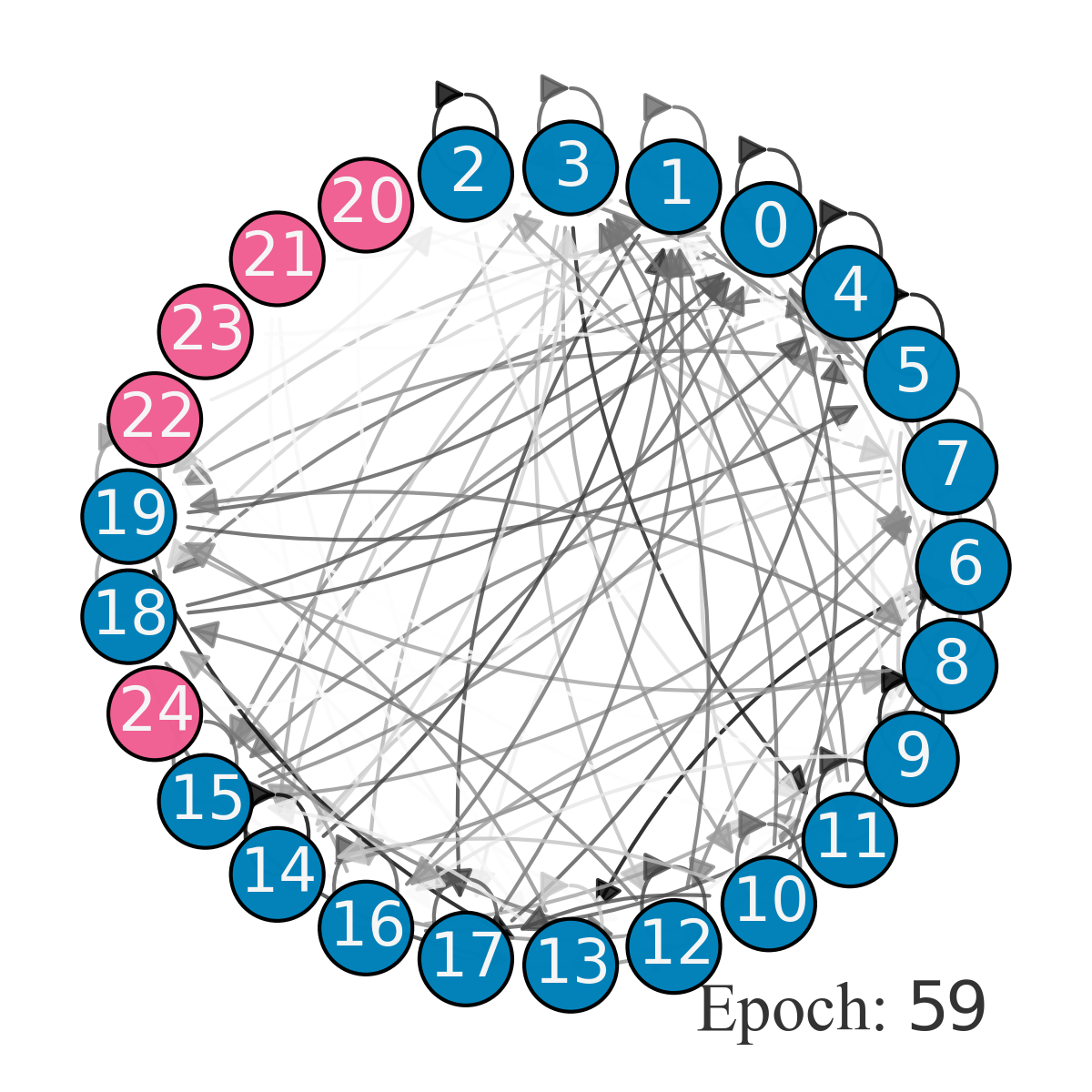}}
    \subfloat{%
        \includegraphics[width=0.25\linewidth]{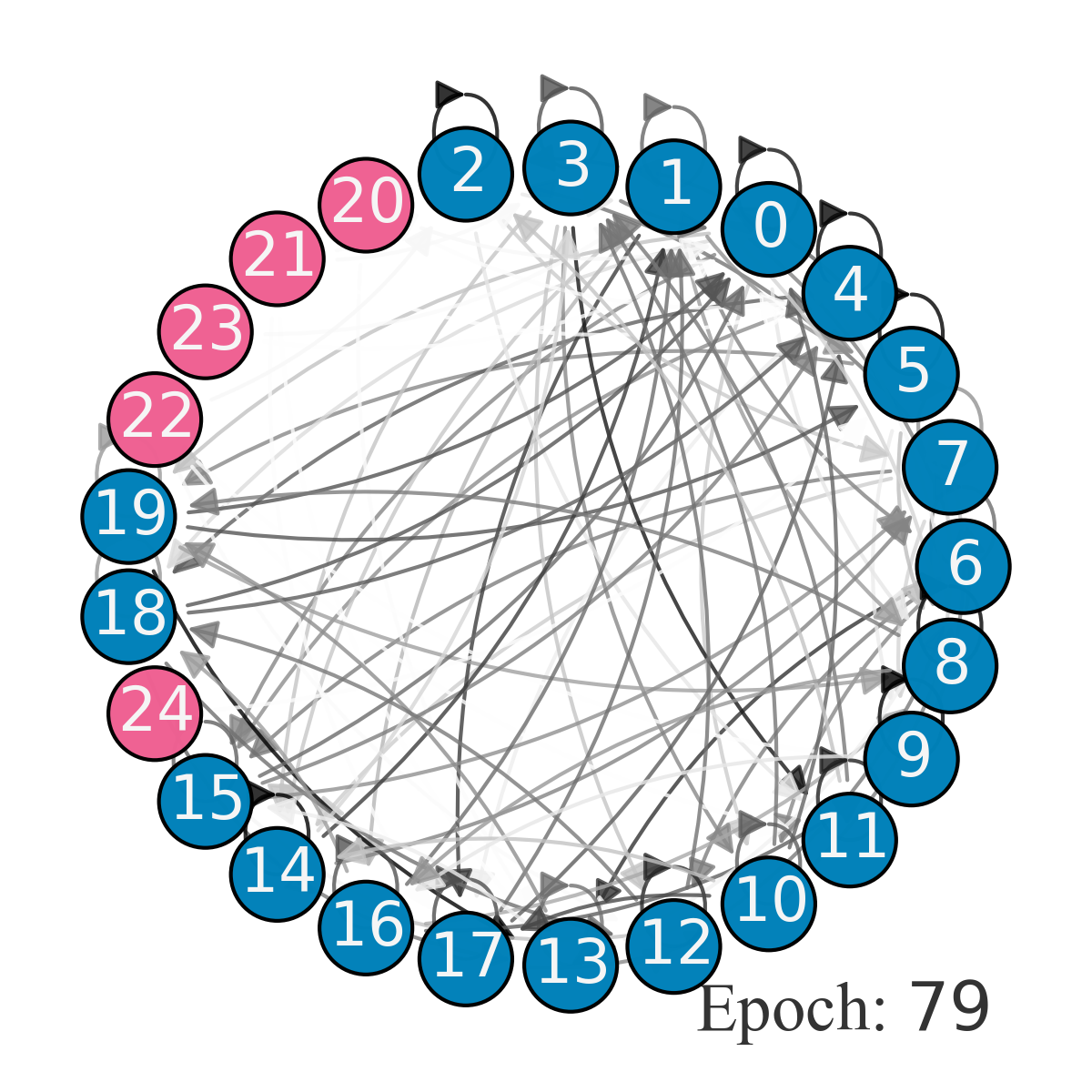}}
    \subfloat{%
        \includegraphics[width=0.25\linewidth]{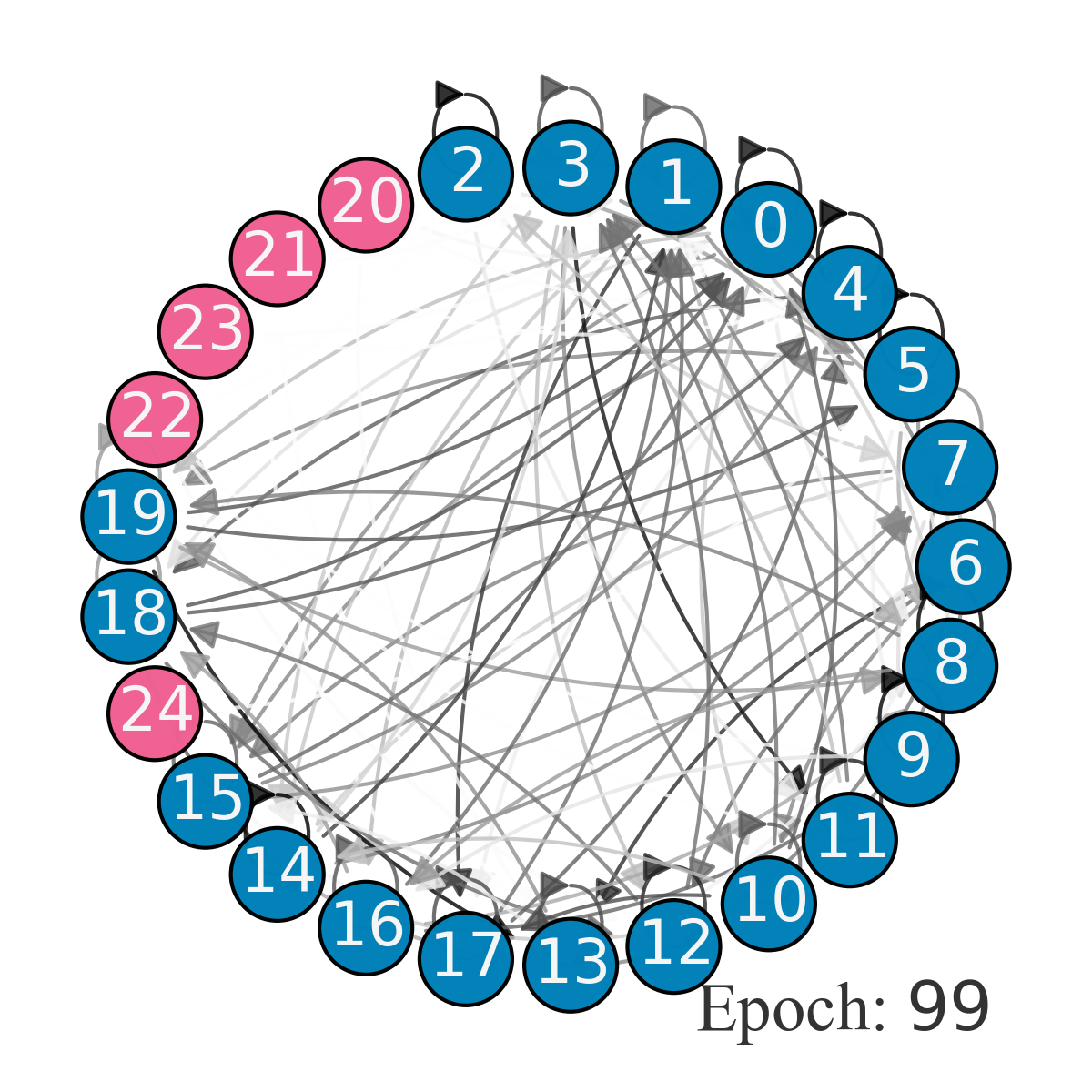}}
    \caption{An illustration of connections (marked by directed lines) between workers of MnistNet trained on EMNIST with 5 attackers (marked by pink nodes, indexed by 20 - 24) and 20 normal workers (marked by blue nodes, indexed by 0 - 19). The shade of the line (white to black) represents the confidence weight of this connection. It can be seen that initially, all lines were gray, indicating: 1. workers did not fully trust others at the beginning, and 2. all connections were weighted equally. However, at epoch $20$, it is clear that all malicious actors were already departed from the main network. After that, all connections became stable for the rest of $80$ epochs.}
    \label{fig:dts-role} 
\end{figure*}

\subsection{Robustness}
In reality, where is full of suspect peers in a p2p network that is ignored by the community, conventional centralized/decentralized FL would result in performance catastrophe. For verification, we conduct a series of controllable experiments to demonstrate that our proposed DeFTA could tolerate more than 66\% malicious peers in the whole network, whereas others methods fail. 

We make comparisons with centralized FL (\textit{i.e.,} CFL-S), decentralized FL (\textit{i.e.,} DeFTA w/o. DTS), and our proposed DeFTA on six datasets and models. Specifically, we select multiple workers to intentionally send malicious data to others in multiple tasks, where malicious data are generated simply by adding random noises to the global model. To keep the partitioned datasets and their respective data distribution unchanged, we fix the number of normal workers to $20$, and therefore introduce newly joined workers as malicious actors. The complete results are shown in Table~\ref{tab:dts}.

As Table~\ref{tab:dts} illustrated, only 1 malicious actor in the system is already capable of failing the entire model training process of both CFL-S and DeFL. Accordingly, we could argue that even CFL-F would also fail since the sampling of CFL-S would possibly ignore the malicious actors. On the contrary, even $40$ malicious actors in the system cannot disturb the training process of DeFTA, which with no doubt validates the effectiveness of DTS. Furthermore, to better understand the role of DTS, the connections between workers of MnistNet trained on EMNIST dataset with 5 malicious actors are visualized in Figure~\ref{fig:dts-role}. The reason for choosing MnistNet and EMNIST as the visualization sample is that they both have a relatively small model collocated with a considerably large dataset, which is considered the most common scenario for practical large-scale FL.

In detail, Figure~\ref{fig:dts-role} shows that initially (\textit{i.e.} epoch $0$), all directed lines between workers were gray and of the same shader. This implies that 1. each worker only partially trusted its peers at the beginning, and 2. all connections of a worker were weighted equally. However, at epoch $20$, all connections to malicious workers have already been faded away, suggesting that DTS successfully eliminates those malicious actors from the main network. After that, all connections become stable for the rest of the epochs, meaning the transition matrix $P$ in DeFTA was converged favorably.

\subsection{AsyncDeFTA}
\label{subsec:exp-asyncdefta}
Most of mentioned decentralized FL methods ignore the asynchronous runtime, whereas our proposed asyncDeFTA accounts for this. The comparison results of asyncDeFTA and DeFTA are illustrated in Table~\ref{tab:async-defta}.

From Table~\ref{tab:async-defta}, it can be observed that the final accuracy/PPL of asyncDeFTA is slightly lower/higher than DeFTA with the same setting, and its variance is also larger. However, with longer train process, \textit{i.e.,} AsyncDeDTA-L (500 epochs for AsyncDeDTA), its best model accuracy/PPL remained roughly equal to DeFTA. The reason for this situation is that in asyncDeFTA, there are few workers that trained so fast that they have already finished training while other workers are still at early epochs. Consequently, these fast workers can only aggregate models from their peers that were highly immature, resulting in a lower model performance when they finished training. However, slower workers can aggregate mature models from fast workers so their models are not affected after training. Therefore, we argue that if fast workers continue to train, the accuracy/perplexity of asyncDeFTA should gradually approach DeFTA, and eventually become the same when all workers' local models converge to their respective local minimal.

\begin{table}[htbp]
  \centering
  \caption{Accuracy/PPL comparsion of synchronous DeFTA and asynchronous DeFTA with $20$ workers}
  \resizebox{\columnwidth}{!}{%
    \begin{tabular}{cccc}
    \toprule
    Model,Dataset & DeFTA(\%)$\uparrow$ & AsyncDeFTA(\%)$\uparrow$ & AsyncDeFTA-L(\%)$\uparrow$ \\
    \midrule
    MLP,MNIST & \textit{97.19$\pm$0.5} & 94.05$\pm$2.3 & \textbf{97.14$\mathbf{\pm}$0.7} \\
    MnistNet,FMNIST & \textit{88.28$\pm$1.9} & 86.17$\pm$4.3 & \textbf{88.96$\mathbf{\pm}$1.3} \\
    MnistNet,EMNIST & \textit{79.24$\pm$2.7} & 73.83$\pm$6.1 & \textbf{79.77$\mathbf{\pm}$2.1} \\
    CNNCifar,Cifar10 & \textit{48.97$\pm$9.5} & 29.41$\pm$12.7 & \textbf{48.13$\mathbf{\pm}$15.1} \\
    VGG,Cifar10 & \textit{64.94$\pm$10.4} & 53.17$\pm$12.9 & \textbf{64.71$\mathbf{\pm}$10.5} \\
    ResNet,Cifar10 & \textit{61.74$\pm$11.2} & 38.9$\pm$16.3 & \textbf{61.53$\mathbf{\pm}$8.1} \\
    ResNet,Cifar100 & \textit{15.72$\pm$5.4} & 8.13$\pm$5.7 & \textbf{15.67$\mathbf{\pm}$5.4} \\
    \midrule
          & DeFTA$\downarrow$ & AsyncDeFTA$\downarrow$ & AsyncDeFTA-L$\downarrow$ \\
    \midrule
    Transformer,Wikitext-2 & \textit{1.201$\pm$0.001} & 1.226$\pm$0.007 & \textbf{1.204$\mathbf{\pm}$0.004} \\
    \bottomrule
    \end{tabular}%
  }
  \label{tab:async-defta}%
\end{table}%

\section{Discussion}
\label{sec:discussion}
\subsection{Privacy of DeFTA}
Recalling the works \cite{zhu2020deep, zhao2020idlg} that reconstruct the original raw training dataset from intermediate data during the model training process, it is clearly possible for the central server to implement this type of attack in the original centralized FL, leading to a potential single-point-of-failure problem. Hence, given that workers in DeFTA are connected and exchange models directly with their peers, a concern naturally arises: Is this type of attack also possible in DeFTA? The answer is \textit{no}. In order to implement such an attack, the attacker needs to have access to both the victim's model before and after the training. However, in DeFTA, each worker's model before its training is randomly aggregated using its local confidence scores and weights, which are unknown to other workers. Hence, the victim's model before the training cannot be retrieved by the attacker. Furthermore, with asyncDeFTA, each worker is capable of having its own training strategy, making the attack extremely hard to implement.

\subsection{Bootstrapping}
Bootstrapping helps the p2p network evolve by introducing newly joined workers to others, and is crucial for any p2p application. In our case, DeFTA can mainly have two options for bootstrapping the network:
\begin{enumerate}
    \item dedicated bootstrapping servers can be open-sourced and maintained by the community. Each of them provides information (\textit{i.e., } IP address, port, worker status, etc.) about every worker in the network.
    \item dedicated bootstrapping servers can be open-sourced and maintained by the community. Each of them only offers information about a few workers and is regarded as a network entry point. Workers in the network, On the other hand, collaborate to maintain a distributed hash table that holds information about every worker in the network.
\end{enumerate}
In fact, connections in DeFTA do not have to be real. For example, one blogger and his followers can be considered connected.

\subsection{Obtain the Global Model}
In centralized FL, the global model can be easily obtained from the central server during the training process. However, in DeFTA, since the system is fully decentralized, acquiring a global model is a little more complicated, commonly requiring the following steps: 1. connecting to $n_k$ peers who have never been connected before, 2. aggregating peers' models $w_k$ and saving it locally, 3. repeat $m$ times, where $m$ varies depending on demands. After sampling, the averaged model of all sampled models $\sum_{k}^{m} \frac{n_k}{\sum n} w_k$ represents the stable global model in DeFTA.

\subsection{Limitation of current DTS}
By assuming workers' peers are not reliable by default, DTS effectively enables robust decentralized FL with large numbers of malicious actors. However, such trustless conduct of DTS will unfortunately also cut off connections between normal workers if their datasets are not very relevant (i.e., their data distributions are notably different). Consequently, for DeFTA to operate optimally, a peer selection strategy that selects workers with similar local dataset features as peers is required. Fortunately, in practice, similar peers can generally be found by either prior knowledge or exhaustive trial (\textit{i.e.,} if the connection is cut off, connect to another one).

\subsection{Perpetual Training}
In practice, DeFTA and asyncDeFTA, like other p2p systems, can inherently support perpetual training, meaning the training process can run indefinitely. This would effectively transfer the control of models from the central server to workers, as well as provide other benefits like better ownership. Furthermore, It also assures that fast workers in asyncDeFTA do not suffer from information disadvantages (\textit{i.e.,} receiving immature models~\ref{subsec:exp-asyncdefta}). However, continuous model training can consume extensive computational and battery resources, especially for edge devices. To reduce such consumption, each worker can configure the frequency of model training respectively (\textit{i.e.,} one model training per 24 hours).

\section{Conclusions and future work}
\label{sec:conclusions}
In this paper, a comprehensive and self-contained decentralized FL framework, named DeFTA, is proposed to address various fundamental problems in both centralized FL and decentralized FL. It is intended to be a drop-in replacement for centralized FL, instantly bringing better security, scalability, and fault-tolerance to downstream applications while maintaining maximum compatibility for existing algorithms for centralized FL. Specifically, peer-to-peer architecture is firstly adopted to achieve decentralization and thereby mitigating the single-point-of-failure problem in centralized FL. Next, the outdegrees of workers are taken into account in the model aggregation formula, with a theoretical performance analysis ensuring a comparable model performance with centralized FL. Moreover, by taking advantage of the p2p topology, DTS is proposed to pessimistically assist the system in quickly and accurately identifying and eliminating malicious actors in the network, hence contributing to DeFTA's robustness and reliability. Finally, an asynchronous variant of DeFTA is proposed for practical usage. Extensive experiments on six datasets and six models validate that both asyncDeFTA and DeFTA are capable of achieving high performance while keeping the system decentralized, resilient, and trustless. However, there are still certain weaknesses to be aware of: 1. the convergence rate of asyncDeFTA is not analyzed in this paper due to the intricacy of problem modeling. 2. The performance gap between DeFTA and FedAvg has not yet been closed to its full potential. 3. The incentive rules in current DTS are rather straightforward. Solving these problems will be our next step.


%





\ifCLASSOPTIONcaptionsoff
  \newpage
\fi

\vfill


\end{document}